\documentclass[ejs,noshowframe]{imsart}

\RequirePackage{amsthm,amsmath,amssymb}
\RequirePackage[comma, sort&compress]{natbib}
\RequirePackage[colorlinks,citecolor=blue,urlcolor=blue]{hyperref}
\usepackage{mathabx}
\usepackage{mathtools}
\usepackage{enumitem}
\usepackage{varwidth}
\usepackage[linesnumbered, ruled, noline]{algorithm2e}

\startlocaldefs
\newtheorem{theorem}{Theorem}
\newtheorem{assumption}{Assumption}

\newtheorem{lemma}[theorem]{Lemma}

\theoremstyle{remark}

\newtheorem{claim*}{Claim}
\DeclareMathOperator*{\argmax}{arg\,max}

\DeclareMathOperator*{\sgn}{\text{sgn}}

\numberwithin{equation}{section}
\numberwithin{theorem}{section}
\numberwithin{example}{section}
\numberwithin{figure}{section}
\renewcommand{\theenumi}{\roman{enumi}}

\newcommand{\cR}{\mathcal{R}}
\newcommand{\cA}{\mathcal{A}}
\newcommand{\cM}{\mathcal{M}}
\newcommand{\bz}{\mathbf{z}}
\newcommand{\btheta}{{\boldsymbol \theta}}

\newcommand{\lfsr}{\text{\emph{lfsr}}}
\newcommand{\BHdir}{\text{BH}_{\text{dir}}}
\newcommand{\STSdir}{\text{STS}_{\text{dir}}}
\newcommand{\dBHdir}{\text{dBH}_{\text{dir}}}

\newcommand{\FDRdir}{\text{FDR}_{\text{dir}}}
\newcommand{\FDRdirhat}{\widehat{\text{FDR}}_{\text{dir}}}
\newcommand{\FDR}{\text{FDR}}
\newcommand{\ETD}{\text{ETD}}
\newcommand{\assumpref}[1]{Assumption~\ref{assump:#1}}
\newcommand{\assumpsref}[1]{Assumptions~\ref{assump:#1}}
\newcommand{\assumpssref}[1]{\ref{assump:#1}}
\newcommand{\figref}[1]{Figure~\ref{fig:#1}}
\newcommand{\figsref}[1]{Figures~\ref{fig:#1}}
\newcommand{\figssref}[1]{\ref{fig:#1}}

\newcommand{\secref}[1]{Section~\ref{sec:#1}}

\newcommand{\algref}[1]{Algorithm~\ref{alg:#1}}
\newcommand{\appref}[1]{Appendix~\ref{app:#1}}

\newcommand{\lemref}[1]{Lemma~\ref{lem:#1}}

\newcommand{\thmref}[1]{Theorem~\ref{thm:#1}}

\newcommand{\thmsref}[1]{Theorems~\ref{thm:#1}}
\newcommand{\thmssref}[1]{\ref{thm:#1}}
\newcommand{\tabref}[1]{Table~\ref{tab:#1}}

\endlocaldefs

\begin{document}
\begin{frontmatter}
\title{Adaptive Procedures for Directional False Discovery Rate Control}
%\title{A sample article title with some additional note\thanksref{t1}}
\runtitle{Adaptive directional FDR control}
%\thankstext{T1}{A sample additional note to the title.}

\begin{aug}
%%%%%%%%%%%%%%%%%%%%%%%%%%%%%%%%%%%%%%%%%%%%%%%
%% Only one address is permitted per author. %%
%% Only division, organization and e-mail is %%
%% included in the address.                  %%
%% Additional information can be included in %%
\author{\fnms{Dennis}~\snm{Leung}\ead[label=e1]{dennis.leung@unimelb.edu.au}}
\and
\author{\fnms{Ninh}~\snm{Tran}\ead[label=e2]{ninht@student.unimelb.edu.au}}
%%%%%%%%%%%%%%%%%%%%%%%%%%%%%%%%%%%%%%%%%%%%%%
%% Addresses                                %%
%%%%%%%%%%%%%%%%%%%%%%%%%%%%%%%%%%%%%%%%%%%%%%
\address{School of Mathematics and Statistics\\
University of Melbourne, Victoria, Australia\\
\printead{e1,e2}}
\runauthor{Leung and Tran}
\end{aug}

\begin{abstract}
In multiple hypothesis testing, it is well known that  adaptive procedures can enhance  power via incorporating information about the number of true nulls present. Under independence, we establish that  two adaptive false discovery rate (FDR) methods, upon augmenting sign declarations,   also offer \emph{directional} false discovery rate ($\FDRdir$) control in the strong sense. Such $\FDRdir$ controlling properties are appealing, because adaptive procedures have the greatest potential to   reap substantial gain in power  when the underlying parameter configurations contain little to no true nulls, which are precisely settings where the $\FDRdir$ is an arguably  more meaningful error rate to be controlled  than the FDR.
\end{abstract}



\end{frontmatter}
%%%%%%%%%%%%%%%%%%%%%%%%%%%%%%%%%%%%%%%%%%%%%%
\section{Introduction} \label{sec:intro}

Consider independent observations $z_1,\dots,z_m \in \mathbb{R}$, where each ``$z$-value'' $z_i$ is a noisy measurement of an effect parameter $\theta_i \in \mathbb{R}$. We suppose $m$ is quite large, and will use the notational  shorthand $[m] \equiv \{1, \dots, m\}$ in the sequel. For testing the multiple point null hypotheses
%For a large $m$, \citet{benjamini1995controlling}  introduced what has now become a  standard  paradigm for exploring the  prominent effects, which is to test the point null hypotheses
\begin{equation} \label{ptnull}
H_i: \theta_i = 0, \quad i \in [m],
\end{equation}
\citet{benjamini1995controlling} proposed  their now-celebrated BH procedure to control 
 the \textit{false discovery rate}  (FDR), 
 \[
 \mathbb{E} \Bigg[ \frac{\sum^m_{i=1} \textbf{1}(\text{$H_i$ is rejected and $\theta_i = 0$})}{1 \vee \sum^m_{i=1} \textbf{1}( \text{$H_i$ is rejected} ) }\Bigg],
 \] 
below a target level $q \in (0,1)$.
However, many statisticians, such as  \citet{tukey1962future, tukey1991philosophy}  and \citet{gelman2000type}, consider  testing the point nulls in \eqref{ptnull}   futile, because the effects in reality, however small, are rarely exactly zero. Instead, they  argue that one should test the direction/sign of the effect by declaring either $\theta_i > 0$ or $\theta_i < 0$ as a discovery, or making no declaration about $\theta_i$ at all if there is insufficient evidence to support either direction. Under this  new paradigm, a generic  discovery procedure consists of two components:   

\begin{enumerate}[label=(\roman*)]
\item  $\cR \subseteq  \{1, \dots, m\}$, the set of rejected indices for which sign declarations (discoveries) are made,  and 
\item $\widehat{sgn}_{i}$, the positive or negative sign declared for each $i \in \cR$. (Note that $\widehat{sgn}_{i} \neq 0$.)
\end{enumerate}
We denote such a procedure, or its associated decisions, by $(\widehat{sgn}_{i})_{i \in \cR}$. Its error rate analogous to the FDR  is the \emph{directional} false discovery rate ($\FDRdir$), defined  as 
%\begin{multline*} \label{FDRdir}
\begin{equation} \label{FDRdir_def}
\FDRdir \Bigl[(\widehat{sgn}_{i})_{i \in \cR}\Bigr]
\equiv \mathbb{E}\left[\frac{ \sum_{i \in \cR} {\bf 1}(\text{sgn}(\theta_i) \neq  \widehat{sgn}_{i} )|}{1 \vee |\cR|} \right], 
\end{equation}
%\end{multline*}
and its power  can be  measured by the expected number of true discoveries (ETD)
\[
\ETD \Bigl[(\widehat{sgn}_{i})_{i \in \cR}\Bigr]  \equiv \mathbb{E}\left[ \sum_{i \in \cR} {\bf 1}(\text{sgn}(\theta_i) =  \widehat{sgn}_{i} )\right],
\]
where for any $x \in \mathbb{R}$, $\text{sgn}(x) \equiv \textbf{1} (x > 0)-\textbf{1} (x < 0)$. This paper treats the control of  $\FDRdir$  for the sign discoveries of the $\theta_i$'s; we don't  exclude the possibility that some effect parameters can indeed be zero, so a false discovery amounts to declaring $\theta_i <0$ when the truth is $\theta_i \geq 0$, or vice versa.\footnote{Some  works call \eqref{FDRdir_def} the \emph{mixed} directional false discovery rate, to emphasize the  two types of errors involved: those from declaring any sign at all for $\theta_i$ when $\theta_i = 0$ and those from declaring an opposite sign for $\theta_i$ when $\theta_i \neq 0$. }

Methods for controlling the $\FDRdir$  are surprisingly scant for the simple testing problem above. To our best knowledge, under  the independence among $z_1, \dots, z_m$ and some standard  assumptions  (\assumpref{symm} and \assumpssref{mlr} below), the only known procedure  that can provably control the $\FDRdir$ under a target level $q\in (0, 1)$ in the strong sense, i.e. irrespective of the configuration of $\theta_i$'s, is what we call the \emph{directional} Benjamini and Hochberg ($\BHdir$) procedure  proposed in \citet[Definition 6]{benjamini2005false}. The $\BHdir$ procedure first decides on the set of rejected indices among $[m]$ by applying the standard  BH procedure at level $q$ to the  two-sided $p$-values 
constructed from $z_1, \dots, z_m$, and then declares the sign of each rejected $\theta_i$ as $\text{sgn}(z_i)$.   \citet[Procedure 6]{guo2015stepwise} proposed another procedure almost  identical to the $\BHdir$,  except that the screening  BH step is applied at level $2q$ instead of $q$. This latter procedure, however,  can only control the $\FDRdir$ under its intended target level $q$ when all $\theta_1, \dots, \theta_m$ are nonzero. \citet[Procedures 7-9]{guo2015stepwise}, \citet{zhao2018controlling}  and \citet{guo2010controlling}   consider extensions of the current problem that involve either specific patterns of dependence  or  multidimensional directional decisions.

We will expand the repertoire of available methods for $\FDRdir$ inference with strong  theoretical guarantee. In the FDR literature,  it is known that \emph{adaptive}  methods \citep{benjamini2006adaptive, benjamini2000adaptive} that  incorporate a data-driven estimate of the proportion of true nulls 
\begin{equation} \label{null_prop_def}
\pi \equiv \frac{|\{i: \theta_i = 0\}|}{m}
\end{equation}
 into their procedures have the potential to improve upon the power offered by the vanilla BH procedure. Using martingale arguments, we prove that under independence, two  adaptive methods can also provide strong $\FDRdir$ control upon the augmentation of sign declarations. The first is  \cite{storey2004strong}'s adaptive FDR procedure, which can be seen as an adaptive variant of the BH procedure  (\secref{Storeydir}). The second is a specific  procedure belonging to  a more recent line of methods driven by a technique called ``data masking'' first introduced in \citet{lei2018adapt} (\secref{zdirect}). We also numerically demonstrate  their competitive power performances in \secref{simul}.

Adaptive procedures that can offer $\FDRdir$ guarantees are particularly important for settings whose underlying parameter configurations contain little to no true nulls. Arguably, if most $\theta_i$'s in question are non-zero, the $\FDRdir$ is more meaningful as an error measure compared to the FDR because querying about their signs matters more. Moreover, such ``non-sparse-signal'' settings are precisely those in which adaptive procedures can reap substantial gain in power; see \citet[Section 3.1]{storey2004strong}, where their adaptive FDR procedure demonstrates greater improvements in power over the BH procedure as $\pi$ decreases.  

\subsection{Notation and assumptions}
 
For $a, b \in \mathbb{R}$, we let $a \wedge b \equiv \min(a, b)$ and $a \vee b \equiv \max(a, b)$. For any two subsets $\mathcal{A}, \mathcal{B} \subset [m]$, $\mathcal{A}\subsetneq (\supsetneq) \mathcal{B}$ means $\mathcal{A}$ is a strict subset (superset) of $\mathcal{B}$. $U(\cdot; a, b)$ denotes a uniform density on the interval $[a, b]\subseteq  \mathbb{R}$. $\mathbb{E}_{{\boldsymbol \theta}}[\cdot]$ means a (frequentist) expectation with respect to fixed values of $\theta_1, \dots, \theta_m$. For each $i \in [m]$ and given $\theta_i = \theta$,  $F_{i,\theta}(\cdot)$ denotes the  distribution function of $z_i$ with density $f_{i,\theta}(\cdot) \equiv F_{i,\theta}'(\cdot) > 0$ with respect to the Lebesgue measure on $\mathbb{R}$ (so $F_{i,\theta}(\cdot)$ is implicitly assumed to be smooth and strictly increasing). $\Phi(\cdot)$ and $\phi(\cdot)$ denote the standard normal distribution and density functions, respectively. Additionally, the following assumptions will be made for the two  main theoretical results in this paper (\thmsref{Storeydircontrol} and \thmssref{zdirectcontrol}):

\begin{assumption}\label{assump:symm}
    The null distribution of $z_i$ is known and symmetric around zero, i.e. $F_{i,0}(-z) = 1 - F_{i,0}(z)$ and $f_{i,0}(-z) = f_{i,0}(z)$ for any $z \in \mathbb{R}$. 
\end{assumption}

\begin{assumption}\label{assump:mlr}
The family of densities $\{f_{i,\theta}(\cdot)\}_{\theta \in \mathbb{R} }$ satisfies the monotone likelihood ratio (MLR) property, i.e. for any given $\theta < \theta^* $ and $z < z^* $, 
$
  \frac{f_{i,\theta^*}(z)}{f_{i,\theta}(z)} \leq    \frac{f_{i,\theta^*}(z^*)}{f_{i,\theta}(z^*)} 
$.
\end{assumption}
These are not necessarily the weakest assumptions for our theorems to hold, but are standard enough so as not to distract from the key ideas of the proofs.  Essentially,  \assumpref{mlr} guarantees that  $z_i$ becomes ``stochastically larger'' as $\theta_i$ increases, and 
two examples satisfying  \assumpref{mlr} are the normal distributions $N(\theta,\sigma_i^2)$ for a fixed variance $\sigma_i^2$ and the noncentral $t$-distributions $NCT(\theta,v_i)$ for a fixed degree $v_i$ \citep[Section 3]{kruskal1954}. Moreover, the symmetry condition on  $F_{i,0}$ in \assumpref{symm} is not crucial; it is included primarily to streamline our presentation. If this condition isn't satisfied, instead of $z_i$, one can alternatively consider the transformed statistic $\Phi^{-1}(F_{i,0}(z_i))$ whose density has the form
\begin{equation}\label{transformed_density_of_zi}
\frac{f_{i, \theta}\left(  F^{-1}_{i, 0} \big(\Phi(z)\big) \right) }{
f_{i, 0}\left(F^{-1}_{i, 0} \big(\Phi(z)\big) \right) }\phi(z).
\end{equation} 
 As a function in $z$, \eqref{transformed_density_of_zi} boils down to the symmetric density function $\phi(z)$ when $\theta_i = 0$. The density \eqref{transformed_density_of_zi} also maintains the MLR property, provided that the base density $f_{i, \theta}(\cdot)$ satisfies the MLR property.

\section{Directional control with \cite{storey2004strong}'s adaptive  procedure} \label{sec:Storeydir}

Compute the two-sided $p$-values
\begin{equation} \label{2sided_pv}
p_i \equiv 2 F_{i,0}(- |z_i|), i \in [m]
\end{equation}
 from the null distributions $F_{i, 0}$, and for any $t \in [0,1]$, let
$
\cR (t) \equiv \{i: p_i \leq t\}
$
be the set of rejected indices $i$ defined by $p_i \leq t$; \algref{Storeydir} is a  sign-augmented version of  \cite{storey2004strong}'s adaptive  procedure  for $\FDRdir$ control, which we call the ``$\STSdir$'' for short.

\begin{algorithm}[h]
\caption{The $\STSdir$ procedure at target $\FDRdir$ level $q \in (0,1)$}
\label{alg:Storeydir}
\KwData{ $z_1, \dots, z_m$}
\KwIn{$\FDRdir$ target $q \in (0, 1)$, the two-side $p$-values $\{p_i\}_{i=1}^m$ from \eqref{2sided_pv}, and  $\{\text{sgn}(z_i)\}_{i=1}^m$; }

For a fixed tuning parameter $\lambda \in (0,1)$, compute   $\hat{\pi}(\lambda) \equiv \frac{| \{ i: p_i > \lambda \}|+1}{(1 - \lambda) m}$ as an estimate for $\pi$\;

Compute  $ t^{\lambda}_{q} \equiv \sup \left\{ t \in [0,1] :   \widehat{\FDR}_{\lambda}(t) \leq q \right\}$, where
\[
                    \widehat{\FDR}_{\lambda}(t) \equiv 
                        \begin{cases}
                            \frac{\hat{\pi}(\lambda)mt}{  |\cR(t)| \vee 1   }, &\text{if $t \leq \lambda$} \\
                            1, &\text{if  $t > \lambda$} 
                        \end{cases} ;
\]

Compute the rejection set $\mathcal{R}(t^{\lambda}_{q}) \equiv \{ i : p_i \leq t^{\lambda}_{q} \}$ \;
\KwOut{Sign discoveries  $(\text{sgn}(z_i))_{i \in \mathcal{R}(t^{\lambda}_{q}) }$.}
\end{algorithm}

 $\mathcal{R}(t^{\lambda}_{q})$ is precisely the rejection set produced by the  adaptive procedure in \citet[Theorem 3]{storey2004strong} that  was  proved to  offer strong FDR control for testing the point nulls in \eqref{ptnull}, under the assumptions that the null $p$-values are independent and uniformly distributed.
%, which is automatically satisfied by how the two-sided $p$-values are constructed  in \eqref{2sided_pv}, \assumpref{symm} and the independence among $z_1, \dots, z_m$.  
Roughly speaking, compared to the  BH procedure, which essentially uses 
\begin{equation} \label{BH_FDR_est}
\frac{m t}{|\cR(t)|\vee 1}
\end{equation}
 as an estimate for the FDR incurred by rejecting any $p_i$ below a given threshold $t$, \citet{storey2004strong}'s FDR estimate $ \widehat{\FDR}_{\lambda}(t)$ adjusts \eqref{BH_FDR_est} by the factor $\hat{\pi}(\lambda)$. This factor serves as a conservative estimate for the unknown null proportion $\pi \in [0, 1]$ in \eqref{null_prop_def}. If $\pi$ is close to zero, i.e. there are very few true nulls in the problem,  \citet{storey2004strong}'s procedure could produce substantially more rejections compared to the BH procedure. Our first result is that, by simply augmenting the rejections with sign declarations as in the output of \algref{Storeydir}, the procedure also offers strong $\FDRdir$ control under independence; we note that by the construction of the $p$-values in \eqref{2sided_pv} and \assumpref{symm}, it must be the case that $\text{sgn}(z_i) \neq 0$ for any $i \in \mathcal{R}(t^{\lambda}_{q})$:

\begin{theorem}[$\FDRdir$ control of $\STSdir$] \label{thm:Storeydircontrol}
%Let $R(t^{\lambda}_q) = | \mathcal{R}(t^{\lambda}_q) |$ and $S(t^{\lambda}_q)=| \{ i \in [m] :  p_i \leq t^{\lambda}_q \text{ and } \sgn(z_i) \neq \sgn(\theta_i) \}|$. 
 Under  \assumpsref{symm} and \assumpssref{mlr}, as well as the independence among $z_1, \dots, z_m$, \algref{Storeydir} controls the $\FDRdir$ at level $q \in (0, 1)$, i.e. by letting $\mathcal{S}(t) \equiv \{i : \text{sgn}(\theta_i) \neq  \text{sgn}(z_i)    \text{ and }  i \in \cR(t) \}$.
\begin{equation*}
    \FDRdir \left[ (\text{sgn}(z_i))_{i \in \mathcal{R}(t^{\lambda}_q)} \right] \equiv \mathbb{E}_{\boldsymbol{\theta}} \left[ \frac{ |\mathcal{S}(t^{\lambda}_q)| }{|\cR(t^{\lambda}_{q})| \vee 1} \right]\leq \Big(1 - \mathbb{E}_{\boldsymbol{\theta}}[\lambda^{|\mathcal{S}(1)|}]\Big) q  \leq  q.
\end{equation*}
\end{theorem}

The proof of \thmref{Storeydircontrol} in \appref{Storeydir} relies on the optional stopping time theorem for supermartingales, which extends the original arguments by \citet[Theorem 3]{storey2004strong} based on martingales.
The choice of $\lambda$ entails a bias-variance trade-off in the estimation of $\pi$, and \cite{storey2004strong} fixes $\lambda = 0.5$ in their simulations. Alternatively, similarly to \citet[Section 6]{storey2004strong}, we provided a method in \appref{auto_lambda} for automatically selecting $\lambda$ based on the observed data $z_1, \dots, z_m$ to achieve a balance between bias and variance. Since \thmref{Storeydircontrol} is only valid on the premise of a fixed $\lambda$, using the provided method of automatically selecting $\lambda$ may not guarantee  $\FDRdir$ control. Regardless, our simulations in \secref{simul} demonstrate that this data-driven choice of $\lambda$ typically leads to empirically robust control of  $\FDRdir$.

\section{Directional control with  data masking: ZDIRECT} \label{sec:zdirect}

Early FDR   methods, such as  \citet{storey2004strong}'s adaptive procedure covered in \secref{Storeydir}, process data in a pre-determined manner to decide on the  set of rejected hypotheses.
% Recently, the   ``data-masking" technique proposed by \citet{lei2018adapt} has opened up new avenues for FDR testing. 
In contrast to the conventional belief that the design of a testing procedure should not be affected by the observed data's patterns to avoid data snooping, \citet{lei2018adapt} recently showed that, as long as the data are initially ``masked'' as a trade-off, one can iteratively interact with the gradually revealed data in a legitimate way to devise  valid FDR procedures. The flexibility of this masking technique  not only allows one to adapt to  information about the null proportion  in \eqref{null_prop_def} implied by the data, but also side information provided by suitable external covariates that are present; see \citep{lei2021,tian2021,leung2022zap,chao2021,yurko2020} for a series of follow-up works, including the ZAP (finite) algorithm proposed by one of the present authors  \citep[Algorithm 2]{leung2022zap}. 

Motivated by the construct of ZAP, 
we propose a  masking algorithm that is augmented with sign declarations for the rejected $\theta_i$'s, and is proven to provide strong    $\FDRdir$ control.
To facilitate our presentation, for each $i \in [m]$, we first define \begin{equation} \label{ui_def}
u_i \equiv F_{i,0}(z_i),
\end{equation}
the  probability integral transform  of $z_i$ under the null  that takes values in the unit interval $(0, 1)$,
as well as its \emph{reflection}
 \begin{equation}\label{reflection}
\widecheck{u}_i \equiv (0.5 - u_i ) I(u_i \leq 0.5) + (1.5 - u_i)  I(u_i > 0.5) 
\end{equation}
 about the midpoint $0.25$  of the left sub-interval $(0, 0.5)$, or about  the midpoint at $ 0.75$ of the right sub-interval $(0.5, 1)$, depending on whether $u_i \leq 0.5$ or $u > 0.5$. Moreover, we  let
 \begin{equation} \label{u_prime_def}
u_i' \equiv
  \begin{cases} 
u_i\wedge \widecheck{u}_i    & \text{if } u_i , \widecheck{u}_i  \in [0, 0.5]\\
u_i\vee \widecheck{u}_i    & \text{if } u_i , \widecheck{u}_i  \in (0.5, 1)
  \end{cases},
\end{equation}
which, for any given $i$, is the value between $u_i$ and $\widecheck{u}_i$ closer to the endpoints of the unit interval. In particular, $u_i'$ only provides  partial knowledge about the value of $u_i$: If $u_i' \leq 0.5$,  $u_i$ may be either $u_i'$ or $0.5 - u_i'$; if $u_i' >0.5$, $u_i$ may be either $u_i'$  or $1.5 - u_i'$. 

With these preparations, we are now in a position to describe our  iterative testing procedure for the directional inference of $\theta_1, \dots, \theta_m$ with a target $\FDRdir$ level $q \in (0, 1)$, whose  steps are listed out in \algref{zdirect}; we call it ``ZDIRECT'' as it  interacts with the $z$-values via  their one-to-one transformations $u_i$'s  for $\FDRdir$ control. Essentially, the algorithm iterates through  steps $t =0, 1, \dots$ to construct a strictly decreasing  sequence of subsets 
\begin{equation*} \label{decreasing_masked_sets}
[m] = \mathcal{M}_0 \supsetneq \mathcal{M}_1   \supsetneq \mathcal{M}_2 \cdots
\end{equation*}
based on the data, and from each $\cM_t$, it forms the candidate ``acceptance'' and ``rejection'' sets
\begin{multline*}
\mathcal{A}_t \equiv \{i: i \in \cM_t \text{ and }  u_i \in (0.25, 0.75)\} \text{ and }\\
 \mathcal{R}_t \equiv \{i: i \in \cM_t \text{ and } u_i\in (0, 0.25]  \cup  [0.75, 1)\};
\end{multline*}
the algorithm terminates at step $\hat{t} = \min\{t:\FDRdirhat (t) \leq q \}$ as soon as  the $\FDRdir$ estimate
\[
\FDRdirhat (t)  \equiv  \frac{1 + |\mathcal{A}_t|}{|\mathcal{R}_t|\vee1}
\]
falls below $q$, and declares the sign discoveries $(\text{sgn}(z_i))_{i \in \mathcal{R}_{\hat{t}}}$.  Since $\mathcal{R}_{\hat{t}} \subset \{i: u_i\in (0, 0.25]  \cup  [0.75, 1)\}$, it must be the case that $\text{sgn}(z_i) \neq 0$ for any $i \in \mathcal{R}_{\hat{t}}$, under \assumpref{symm}. Moreover, the sets $\cM_t$ are  shrunk  in such a way that these two conditions must be respected (line 6 in \algref{zdirect}):
\begin{enumerate}
%\item [(C1)] $\cM_0$ can be any subset in $[m]$ determined based on the ``masked" dataset $\{u_i'\}_{i=1}^m$ only.
\item[(C1)] $\cM_{t +1}$ must be constructed based only on the partial data  $\{\tilde{u}_{i, t}\}_{i \in [m]}$ available at step $t$, where 
 for each $i \in [m]$, we  define 
  \begin{equation} \label{U_tilde}
\widetilde{u}_{i, t}   \equiv 
 \begin{dcases*}
        u_i  & if $i \not\in  \mathcal{M}_t$\\
        u_i'& if $i \in  \mathcal{M}_t$
        \end{dcases*},
\end{equation}
which may reveal the true value of $u_i$ depending on whether $u_i \in \mathcal{M}_t$. Essentially,  any human/computer routine who shrinks $\cM_t$ to $\cM_{t+1}$ can only know that the true value of $u_i$ is one of  two possibilities  if $i$ is still in the ``masked'' set $\mathcal{M}_t $.  
 \item[(C2)] 
$\mathcal{M}_{t+1}$ must be a strict subset of $\mathcal{M}_t$
to ensure  the algorithm  does  terminate. 
\end{enumerate}

\begin{algorithm}[h]
\caption{The ZDIRECT procedure at target $\FDRdir$ level $q \in (0,1)$}
\label{alg:zdirect}
\KwData{ $z_1, \dots, z_m$}
\KwIn{$\FDRdir$ target $q \in (0, 1)$, the initial  set $\mathcal{M}_0 = [m]$
% chosen only based on  $\{u_i'\}_{i \in [m]}$
 ; }

\For{t= 0,1 \dots, }{
 Find the candidate ``acceptance set'' $\mathcal{A}_t \equiv \{i: i \in \cM_t \text{ and }  u_i \in (0.25, 0.75)\}$\;
 Find the candidate ``rejection set''  $\mathcal{R}_t \equiv \{i: i \in \cM_t \text{ and } u_i\in (0, 0.25]  \cup  [0.75, 1)\}$\;
Compute $\FDRdirhat (t)  \equiv  \frac{1 + |\mathcal{A}_t|}{|\mathcal{R}_t|\vee1}
$\;
  \eIf{$\FDRdirhat (t)  > q$}{
%   Update $s_{l,t+1}$ and $s_{r,t+1}$ while respecting the two conditions in \thmref{ZAPfiniteControl}.
%    E.g. Apply \algref{updateThreshold} in \appref{EM_finite_update}
 construct $ \cM_{t+1} \subsetneq  \cM_{t}  $ using only the partially masked data $\{\tilde{u}_{i, t}\}_{i \in [m]}$ from \eqref{U_tilde}\;
   }{
      Set $\hat{t}  = t$; break\;
  }
 }
\KwOut{Sign discoveries  $(\text{sgn}(z_i))_{i \in \mathcal{R}_{\hat{t}}}$.}
\end{algorithm}

 \thmref{zdirectcontrol}, which is proven in \appref{FDRdirpf}, states that ZDIRECT provides strong $\FDRdir$ control under the assumptions in this paper: 
\begin{theorem}[$\FDRdir$ control of ZDIRECT]\label{thm:zdirectcontrol}
Under  \assumpsref{symm} and \assumpssref{mlr}, as well as the independence between $z_1,\dots, z_m$, \algref{zdirect} controls the $\FDRdir$ at level $q \in (0, 1)$. Specifically, if 
%\[
%dFDP(t) = \frac{|\{i : \text{sgn}(\theta_i) \neq  \text{sgn}(z_i)    \text{ and }  i \in \mathcal{R}_{t} \}|}{ |\mathcal{R}_{t}|}
%\]
%\[
%\mathbb{E}[\frac{|\{i : \text{sgn}(\theta_i) \neq  \text{sgn}(z_i)    \text{ and }  i \in \mathcal{R}_{t} \}|}{ |\mathcal{R}_{t}|} ] \leq q
% \]
%denotes the directional false discovery proportion incurred at a given step $t$ and 
the algorithm terminates at step $\hat{t} = \min\{t: \FDRdirhat (t) \leq q\}$, we have
\begin{equation} \label{FDRdircontrol}
\mathbb{E}_{{\boldsymbol \theta}}\left[ \frac{|\{i : \text{sgn}(\theta_i) \neq  \text{sgn}(z_i)    \text{ and }  i \in \mathcal{R}_{\hat{t}} \}|}{1 \vee |\mathcal{R}_{\hat{t}}|} \right] \leq q.
\end{equation}
\end{theorem}

%Essentially, \algref{zdirect} iteratively constructs a sequence of non-increasing candidate rejection sets $\cR_t$ from $\cM_t$, until the ratio $\FDRdirhat (t)$, which estimates the $\FDRdir$ that would be incurred by making the sign declarations $(\text{sgn}(z_i))_{i \in \mathcal{R}_{t}}$, drops below $q$. In particular, the sets $\cM_t$ are  shrunk based on the data in such a way that these two conditions are respected:
%\begin{enumerate}
%%\item [(C1)] $\cM_0$ can be any subset in $[m]$ determined based on the ``masked" dataset $\{u_i'\}_{i=1}^m$ only.
%\item[(C1)] $\cM_{t +1}$ must be constructed based on the partial data  $\{\tilde{u}_{i, t}\}_{i \in [m]}$ available at step $t$ only.   Essentially,  any human/computer routine who shrinks $\cM_t$ to $\cM_{t+1}$ can only know that the true value of $u_i$ is one of  two possibilities  if $i$ is still in the ``masked'' set $\mathcal{M}_t $.  
% \item[(C2)] 
%$\mathcal{M}_{t+1}$ must be a strict subset of $\mathcal{M}_t$
%to ensure  the algorithm  does  terminate. 
%\end{enumerate}
Apart from the final sign declarations,  \algref{zdirect} inherits much of its structure from the ZAP (finite) algorithm in \citet{leung2022zap}, but is situated in the more general ``without-threshold'' framework \citep[Section 6.1]{lei2018adapt} that does not explicitly involve any thresholding functions.  In fact, a quantity like $\FDRdirhat (t)$ has been used as an FDR estimate in \citet{leung2022zap}. 
 To also make sense of it as a suitable $\FDRdir$ estimate for the sign discoveries of $(\text{sgn}(z_i))_{i \in \mathcal{R}_{t}}$, 
 note that an $i \in \mathcal{R}_t$ constitutes a false discovery if either
\begin{equation} \label{FD1}
u_i \in (0, 0.25] \text{ and } \text{sgn}(\theta_i) \geq 0,
\end{equation}
or
\begin{equation}  \label{FD2}
u_i \in [0.75, 1) \text{ and } \text{sgn}(\theta_i) \leq 0.
\end{equation}
By the continuity of $F_{i,0}$, each $u_i$ is  uniformly distributed when $\theta_i = 0$. 
 Suppose we are in an error-prone scenario where all  $\theta_i$'s are  either exactly or very close to zero, so that all $u_i$'s stochastically behave  much like uniform random variables. For  a given $i \in \cM_t$, the event $\{u_i \in (0, 0.25] \}$ and $\{u_i \in (0.25, 0.5] \}$ should be approximately equally likely, so  the set size $|\{i: i \in \cM_t \text{ and } u_i \in (0.25,  0.5]\}|$ is a reasonable estimate of the number of false discoveries by way of \eqref{FD1}. Analogously, $|\{i: i \in \cM_t \text{ and } u_i \in [0.5, 0.75)\}|$ serves as an estimate of the number of false discoveries by way of \eqref{FD2}. As such, $|\mathcal{A}_t|$ makes sense as an  estimate of the number of false discoveries in  $\mathcal{R}_t$, where the additive ``1'' in the numerator of  $\FDRdirhat (t) $ is a theoretical adjustment factor to make it conservative enough. In fact, this concept is akin to how the FDR estimate for the \emph{knockoff filter}  for variable selection in linear regressions \citep{barber2015controlling} can also serve as an $\FDRdir$ estimate when sign declarations are augmented \citep{barber2019knockoff}.

However, our specific methodology for updating $\cM_t$, as described next, considerably differs from existing data masking algorithms. Notably, we rely solely on $z_1, \dots, z_m$ as the available data for our problem, without harnessing external covariate information. This poses a greater challenge for boosting power, but our simulations in \secref{simul} demonstrate that ZDIRECT remains competitive in terms of power for $\FDRdir$ control when compared to other existing methods.

\subsection{Shrinking the masked sets $\cM_t$} \label{sec:oracle}

To achieve  power, we  aim to shrink $\cM_t$ in accordance with  conditions (C1)-(C2) in such a way that \algref{zdirect}  can mimic  the \emph{optimal discovery procedure} (ODP) under a Bayesian formulation\footnote{More precisely, it is an \emph{empirical Bayes} formulation, because $G(\cdot)$ is assumed to be unknown.} of the problem, which imposes the additional assumption that the effects are  random and independently generated by a common prior distribution, i.e.
\begin{equation} \label{prior}
\theta_i   \sim_{i.i.d.} G( \theta), \quad i \in [m],
\end{equation}
for an unknown distribution function $G(\cdot)$. For  a given target level $q \in (0, 1)$, the ODP, denoted by $ (\widehat{sgn}^{ODP}_i)_{i \in \cR^{ODP}}$, is  defined to be the procedure with the properties that
\[
\FDRdir \Bigl [(\widehat{sgn}^{ODP}_{i})_{i \in \cR^{ODP}}\Bigr]  \leq q
\]
and 
\begin{equation*}
\ETD \Bigl[ (\widehat{sgn}^{ODP}_{i})_{i \in \cR^{ODP}}\Bigr] 
 \geq 
\ETD \Bigl[ (\widehat{sgn}_{i})_{i \in \cR}\Bigr]
\end{equation*}
for any procedure $(\widehat{sgn}_i)_{i \in \cR}$  with $\FDRdir \Bigl[ (\widehat{sgn}_{i})_{i \in \cR}\Bigr]  \leq q$, where the expectation operator defining all the $\FDRdir$ and ETD quantities just mentioned is with respect to \emph{both} the randomness of the  data $\{z_i\}_{i \in [m]}$ and that of the parameters $\{\theta_i \}_{i \in [m]}$ under  \eqref{prior} (i.e. different from the frequentist $\mathbb{E}_{{\boldsymbol \theta}}[\cdot]$ operator). In other words, the ODP is the best procedure in terms of maximizing $\ETD$, among all that can control $\FDRdir$ under a target level. 

In order to shrink $\cM_t$ in a manner that \algref{zdirect} can mimic the ODP, we have to  better understand the latter's operational characteristics; to that end, we shall first intuitively grasp what the best course of action for a directional decision maker should be  under the Bayesian assumption  \eqref{prior}. If s/he were
 obliged to unambivalently
declare a non-zero sign for a \emph{specific} $\theta_i$ based on  $z_i$, the optimal strategy   is clearly the one outlined in \tabref{OptSign} based on the posterior probabilities $P(\theta_i < 0 |z_i)$ and $P(\theta_i > 0 |z_i)$, whose associated probability of making a false  discovery can be calculated as
\begin{equation} \label{lfsr}
\lfsr_i = P(\theta_i \leq 0 |z_i) \wedge P(\theta_i \geq 0 |z_i)
\end{equation}
and must be no larger than the probability of false discovery made by any other strategy. 
In the literature, the quantity in \eqref{lfsr} is known as the \emph{local false sign rate} for $i$ \citep[p.279]{stephens2017false}, and  a smaller $\lfsr_i$ suggests higher confidence in the sign declaration for  $\theta_i$ prescribed by  \tabref{OptSign}. Now, if s/he were to make non-zero sign declarations for the largest possible subset of parameters from $\{\theta_i\}_{i \in [m]}$ with $\FDRdir$ control in mind,  the intuition would be to prioritize making sign declarations for those $i$'s with the smallest local false sign rates, each using \tabref{OptSign}'s strategy. This is in fact what the  ODP does,  as stated in \thmref{ODP} below. We note that the ODP is not implementable in practice as the underlying prior $G(\cdot)$ is, by assumption, unknown  for computing the local false sign rates.

\begin{theorem}[Operational characteristics of the ODP under Bayesian formulation]\label{thm:ODP}
Assume the prior in \eqref{prior}, and that conditional on $\{\theta_i \}_{i \in [m]}$, $z_1, \dots, z_m$ are independent with respective distributions $F_{1,\theta_1}, \dots, F_{m,\theta_m}$. 
For  a given level $q \in (0, 1)$, the optimal discovery procedure $ (\widehat{sgn}^{ODP}_i)_{i \in \cR^{ODP}}$ must be such that 
%suppose $ (\widehat{sgn}^{ODP}_i)_{i \in \cR^{ODP}}$ is the optimal discovery procedure (ODP) such that
%\[
%\FDRdir \Bigl [(\widehat{sgn}^{ODP}_{i})_{i \in \cR^{ODP}}\Bigr]  \leq q
%\]
%and 
%\begin{equation*}
%\ETD \Bigl[ (\widehat{sgn}^{ODP}_{i})_{i \in \cR^{ODP}}\Bigr] 
% \geq 
%\ETD \Bigl[ (\widehat{sgn}_{i})_{i \in \cR}\Bigr]
%\end{equation*}
%for any procedure $(\widehat{sgn}_i)_{i \in \cR}$  with $\FDRdir \Bigl[ (\widehat{sgn}_{i})_{i \in \cR}\Bigr]  \leq q$. 
%It  must be true that
\begin{enumerate} [label=(\roman*)]
\item  $\text{lfsr}_i \leq \text{lfsr}_j$  for any $i \in \mathcal{R}^{ODP} \text{ and } j \in  [m] \backslash \mathcal{R}^{ODP}$, and
\item  For each $i \in \cR^{ODP}$, $\widehat{sgn}_i^{ODP}$ is declared in accordance with  \tabref{OptSign}.
\end{enumerate}
%Here, the expectation  defining the $\FDRdir$ and ETD  are with respect to both the data $\{z_i\}_{i \in [m]}$ and random parameters $\{\theta_i \}_{i \in [m]}$ under  \eqref{prior}.
\end{theorem}

\begin{table}[t]
  \centering
    \begin{tabular}{ | c  | p{6cm} |}
    \hline
    If &  Then \\ \hline
    $ P(\theta_i <  0 |z_i ) < P(\theta_i > 0 |z_i )$ & declare $\theta_i > 0$  \\ \hline
    $ P(\theta_i > 0  |z_i ) < P(\theta_i < 0 |z_i )$  & declare $ \theta_i < 0$ \\ \hline
    $ P(\theta_i > 0  |z_i ) = P(\theta_i < 0 |z_i )$ &  declare either $ \theta_i > 0$ or  $\theta_i < 0$\\  \hline
    \end{tabular}
  \caption{Optimal sign declaration strategy for a given $i$.}
  \label{tab:OptSign}
\end{table}

  The proof of  \thmref{ODP} is in \appref{lfsrOMTpf}, which extends the arguments in \citet[Theorem 2.1]{heller2021optimal} on the optimal FDR procedure for the point null testing problem in  \eqref{ptnull}. While the ODP cannot be operationalized in practice, one can make ZDIRECT mimic its characteristic that indices with the smallest local false sign rates get rejected first: At each step $t$, we aim to get rid of exactly one element from the masked set $\mathcal{M}_t $ that  has potentially the largest local false sign rate, since only elements remaining in the next $\mathcal{M}_{t+1} $ may  be  rejected.  In what follows, let 
 \begin{equation*} \label{z_tilde}
 z_i' \equiv  F_{i, 0}^{-1} ( u_i'), \quad \widecheck{z}_i \equiv  F_{i, 0}^{-1} ( \widecheck{u}_i) \quad \text{ and } \quad 
\widetilde{z}_{i, t}   \equiv 
 \begin{dcases*}
       z_i & if $i \not\in  \mathcal{M}_t$\\
             z_i'  & if $i \in  \mathcal{M}_t$
        \end{dcases*},
\end{equation*}
 which convert $u_i'$, $\widecheck{u}_i$ and $\widetilde{u}_{i, t} $ back onto their original $z$-value scale. Specifically,  we estimate the local false sign rates as
\begin{equation}\label{lfsrestimate}
\widehat{\lfsr}_{i, t} =
 \min\left(\frac{ \int_{\theta \leq 0} f_{i,\theta}( z_i')d \hat{G}_t( \theta) }{ \int f_{i,\theta}( z_i') d \hat{G}_t( \theta) }, \frac{ \int_{\theta \geq 0}  f_{i,\theta}( z_i')d \hat{G}_t( \theta) }{ \int  f_{i,\theta}( z_i') d \hat{G}_t( \theta) } \right) \text{ for each } i \in \mathcal{M}_t,
\end{equation}
where $\hat{G}_t (\cdot)$ is an estimate of $G(\cdot)$ based on the partial dataset $\{\tilde{z}_{i, t}\}_{i \in [m]}$ 
, or equivalently, $\{\tilde{u}_{i, t}\}_{i \in [m]}$ from \eqref{U_tilde}.
The index to be unmasked from $\cM_t$ is then 
\begin{equation}\label{stephenindex}
    \hat{i}_t = \argmax_{i \in \mathcal{M}_t} \widehat{\lfsr}_{i,t}, 
\end{equation}
which has the largest estimated local false sign rate; that the estimates in \eqref{lfsrestimate} are evaluated at the $z_i'$'s presumes that any given  masked element in $\mathcal{M}_t$ may come from the rejection set $\mathcal{R}_t$.

Now we describe how we get $\hat{G}_t(\cdot)$ based on the partial  data. Since the  prior $G(\cdot)$ is unknown, we will model it to have a \emph{unimodal} mixture density 
\begin{equation} \label{halfuniform}
g(\cdot ;{\bf w}) = w_0 \delta_0(\cdot) + \sum_{k = -K, \dots, -1, 1, \dots, K} w_k h_k(\cdot) 
\end{equation}
proposed in \cite{stephens2017false}, where ${\bf w} \equiv (w_{-K}, \dots, w_{-1}, w_0,  w_1, \dots, w_K)$ are mixing probabilities that sum to $1$,  $\delta_0(\cdot)$ denotes the delta function at zero, and $h_k$'s are uniform densities of the forms
\[
 h_k( \cdot) = 
  \begin{cases} 
  U(\cdot; 0, a_k ) & \text{if }  k = 1, \dots, K \\
 U(\cdot; a_k , 0)     & \text{if } k = -1, \dots, -K
  \end{cases}
\]
for predetermined endpoints $a_1, \dots, a_K  > 0$ and $a_{-1}, \dots, a_{-K} < 0$. The log-likelihood of \eqref{halfuniform} with respect to the partial data $\{\tilde{z}_{i, t}\}_{i \in [m]}$ at step $t$ can then be computed as 
\begin{equation} \label{logLike}
L_t({\bf w})=
\sum_{i \in [m] }\log \left[w_0  l_{0, i, t} +  \sum_{k = -K, \dots, -1, 1, \dots, K} w_k  l_{k, i, t}   \right] 
\end{equation}
where  $ l_{k, i, t} $ are the likelihoods of the  mixture components of the forms
\[
 l_{0, i, t} = 
  \begin{cases} 
   f_{i,0}(z_i) + f_{i,0}(\widecheck{z}_i)& \text{if }  i \in \mathcal{M}_t \\
   f_{i,0}(z_i)      & \text{if }    i \in [m] \backslash \mathcal{M}_t
  \end{cases}
\] 
and
\begin{equation} \label{component_log_lik}
 l_{k, i, t} = 
  \begin{cases} 
   \int  \left[ f_{i,\theta}(z_i) +  f_{i,\theta}(\widecheck{z}_i) \right ] h_k(\theta) d \theta & \text{if }  i \in \mathcal{M}_t \\
   \int f_{i,\theta}(z_i) \cdot h_k( \theta) d \theta     & \text{if }    i \in [m] \backslash \mathcal{M}_t
  \end{cases} \text{ for } \quad k =-K, \dots, -1,  1, \dots, K.
\end{equation}
When $f_{i, \theta}$ is from the family of  normal distributions $N(\theta, \sigma_i^2)$, the quantities in \eqref{component_log_lik} have closed-form expressions; if $f_{i, \theta}$ is from the family of noncentral $t$-distributions $NCT(\theta, \nu_i)$, methods for approximating  the quantities in \eqref{component_log_lik} are discussed in \appref{compLike}.  
From \eqref{logLike}, the density of $\hat{G}_t$ is taken as
$
\hat{g}_t (\cdot) = g(\cdot; \hat{{\bf w}}_t ),
$
where $\hat{{\bf w}}_t$ solves 
 the penalized maximum likelihood estimation:
\begin{equation}\label{optim}
\max_{\substack{{\bf w} : \\\sum_k w_k = 1, \\ w_k \geq 0  \forall k}} \left[ L_t({\bf w}) +  \sum_{k = -K}^K  ({\lambda_k - 1})\log   w_k \right].
\end{equation}
Above, the last term is a Dirichlet penalty with tuning parameters $\lambda_k >0$. 
%The details of selecting $K$, $a_k$ and the tuning parameters $\lambda_k$ are also deferred to \appref{compLike}.

{\bf Remarks}.  Adaptivity  is implicitly built into our algorithm, since the null probability $w_0$ in our modeling density \eqref{halfuniform}   is the Bayesian analogue of the frequentist null proportion in \eqref{null_prop_def}. By striving to mimic the operational characteristic of the ODP, it also allows adaptivity to other features, such as the asymmetry in the distribution of the $z_i$'s; in the FDR literature, it is well known that \emph{local false discovery rate} approach based on $z$-values can further boosts testing power by leveraging distributional asymmetry \citep{storeyleekdai2007optimal, sun2007oracle}, and the same discussion can carry  over to $\FDRdir$ control with local false sign rates.  We stress that although \eqref{halfuniform} may well be misspecified as a density for the hypothetical prior $G(\cdot)$, strong frequentist $\FDRdir$ control is guaranteed by \eqref{FDRdircontrol} in \thmref{zdirectcontrol}. Moreover,  our choice of it as a working model  carries two main advantages:
\begin{enumerate} [label=(\alph*)]
\item  \emph{Speed}: The iterative updates of $\cM_t$ for most existing FDR data-masking algorithms  are  computationally expensive, as they usually employ beta mixture models that require the EM algorithm \citep{dempster1977maximum} for estimation. On the contrary,  as explained in \citet[Supplementary material]{stephens2017false}, an optimization problem with the form  in  \eqref{optim} is convex and can be solved by fast and reliable interior point methods.  
\item   \emph{Appropriate flexibility}:  \citet{stephens2017false} argues for the plausibility of unimodality  in many real applications since most effects are close to zero while larger effects are decreasingly likely; by increasing the number of components $K$ and expanding the supporting intervals defined by $a_k$, \eqref{halfuniform} can approximate \emph{any unimodal distribution about zero} \citep[p.158]{feller_vol2_1971}. On the other hand, as heuristically discussed in \citet[Section 5]{leung2022zap}, having a too-flexible model could overfit the partial data $\{\tilde{z}_{i, t}\}_{i \in [m]}$, which one is constrained to work with to adhere to data masking, only to underfit the original data $\{z_{i, t}\}_{i \in [m]}$. \citet[Section 3.1.4]{stephens2017false} also advocates for unimodality as a form of ``regularization'' because it prevents density estimates from concentrating in small pockets. Our simulation results in \secref{simul} show that ZDIRECT can still be competitive against other practical methods  even when the unimodality is misspecified.
\end{enumerate}

% simulation section
\section{Simulations} \label{sec:simul}

\subsection{Simulation setups} \label{sec:simulation_setup_prototype}
We simulate $m = 1000$ independent $z_i \sim N(\theta_i, 1)$ values where $\theta_1, \dots, \theta_m$ are independently generated from a mixture density of the form 
\begin{equation} \label{simul_prior}
g(\theta) = w \delta_0 (\theta) + (1 - w)g_1(\theta),
\end{equation}
where $\delta_0(\cdot)$ is the delta function at zero for the nulls, 
and $g_1(\theta)$ is an  ``alternative''  density  which is itself a mixture of two normal components of the form
\[
g_1(\theta) = (1 - v) \phi( \theta + \xi) + v \phi(\theta - \xi).
\]
The simulation parameters controlling different aspects are chosen as follows.
\begin{enumerate} [label=(\alph*)]
\item $w$ (\emph{signal sparsity parameter}): Takes  one of the values in $\{0.8, 0.5, 0.2, 0\}$. Letting $w = 0.8$ renders a setting with approximately $80\%$ of the $\theta_i$'s equal to $0$, and taking $w = 0$ renders a setting with all $\theta_i$'s being non-zero. 
\item $\xi$ (\emph{signal size parameter}): Takes one of the values in  $\{0.5, 1, 1.5, 2, 2.5\}$. It influences the absolute value of an effect $\theta_i$ if it is non-zero.
\item $v$ (\emph{asymmetry parameter}): Takes one of the values in  $\{0.5, 0.75, 1\}$. It controls the proportion of the alternative $\theta_i$'s generated from the positively centered normal component $\phi( \theta - \xi)$; a larger $v$ makes $g(\theta)$ more asymmetric.
\end{enumerate}
How the different combinations of  $w$, $\xi$ and $v$ change the shape of $g(\theta)$ is illustrated in \figref{gplot}; note that many of these $g(\theta)$ are evidently not unimodal.

\begin{figure*}[h]
\centering
\includegraphics[scale = 0.6]{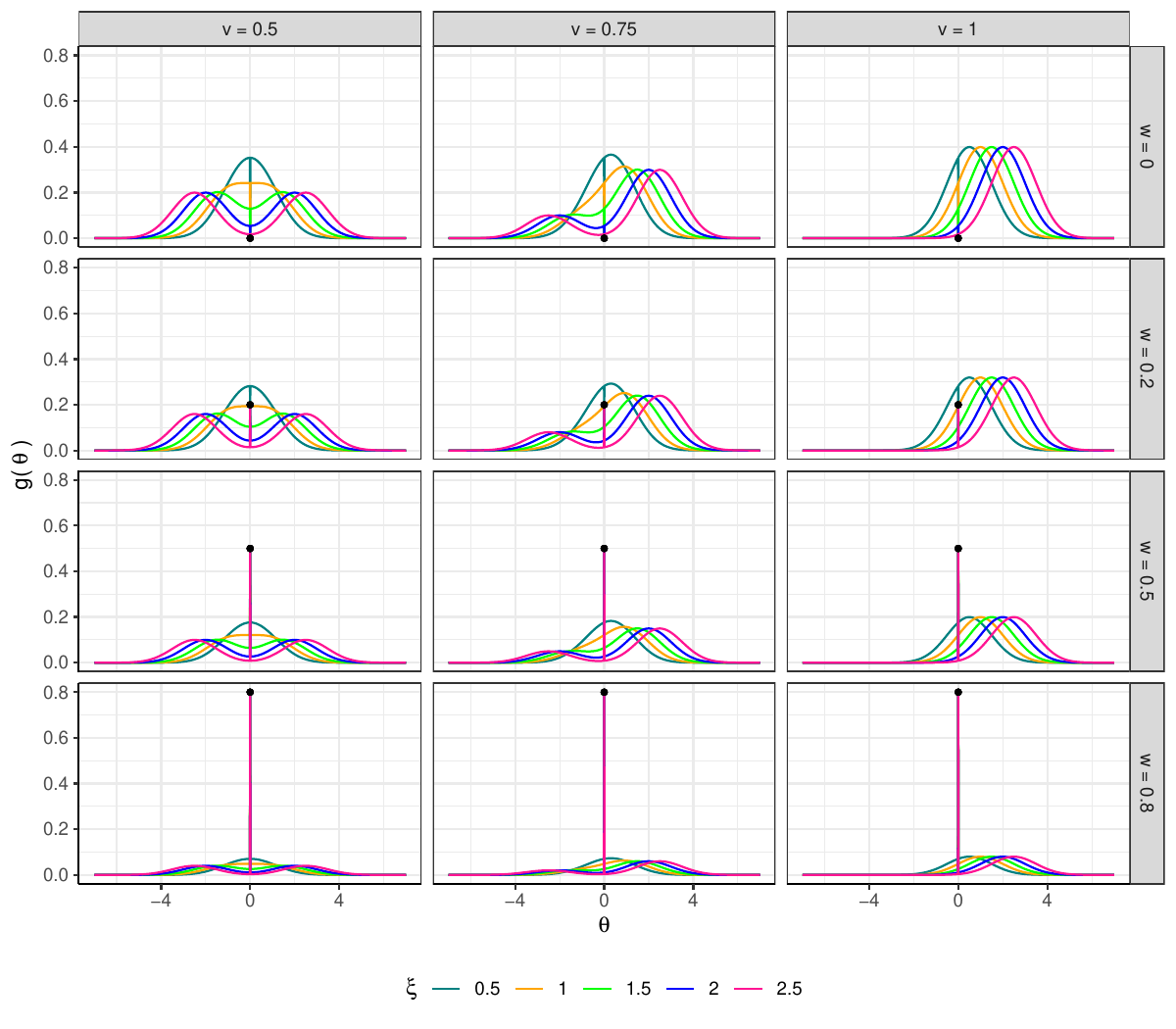}
\caption{
Plots of the effect generating density $g(\theta)$ for every possible combination of $w$, $\xi$ and $v$ considered in \secref{simul}.
 }
\label{fig:gplot}
\end{figure*}

\subsection{Methods compared} \label{sec:methods_compared_sim}

We compare the following seven methods for $\FDRdir$ control:

\begin{enumerate} [label=(\alph*)]
\item BH$_\text{dir}$: The directional BH procedure proposed by \citet{benjamini2005false}. We note that the validity of  BH$_\text{dir}$ is originally proved  under the assumption that the family $\{F_{i, \theta}(\cdot)\}_{\theta \in \mathbb{R}}$ is stochastically increasing, which is implied by the MLR property in \assumpref{mlr} \citep[Lemma 3.4.2 $(ii)$]{lehmann2005testing}.
\item LFSR: A  computationally simpler substitute for the ODP based on the oracle $\lfsr_i$'s in \eqref{lfsr}; inspired by  \citet{sun2007oracle}'s earlier work on optimal FDR control. Its implementation details are deferred to \appref{imple}. The  implementation of the ODP described in \thmref{ODP} entails solving a complex infinite integer problem to compute a rejection threshold for the $\lfsr_i$'s \citep{heller2021optimal}. In comparison, LFSR is  simpler to compute whilst often having comparable power. Like the ODP, LFSR offers $\FDRdir$ control under the Bayesian formulation in \eqref{prior}, but it can only serve as a hypothetical benchmark for other methods since it also requires oracle knowledge of the true generating prior in \eqref{simul_prior}.

\item ASH (``{\bf a}daptive {\bf sh}rinkage'', \citet{stephens2017false}): A procedure that is almost the same as LFSR in its implementation, except that the oracle $g(\cdot)$ from \eqref{simul_prior} is replaced by an estimated $g(\cdot; \hat{\bf w})$ based on the unimodal model in \eqref{halfuniform}, where $\hat{\bf w}$ is a penalized maximum likelihood estimate of ${\bf w}$  with respect to the full data $\{z_i\}_{i \in [m]}$ obtained by the R package \texttt{ashr}; all tuning parameters involved are chosen to be the default described in \citet[Supplementary information]{stephens2017false}. This procedure may risk violating the desired $\FDRdir$ target  if the unimodal density \eqref{halfuniform} is  too far from the true prior density of the $\theta$'s.
\item GR (\citet[Procedure 6]{guo2015stepwise}): The $\FDRdir$ testing procedure mentioned in \secref{intro} that provides (frequentist) control of the $\FDRdir$ under the target level $q$ when all the $\theta_i$'s are non-zero, i.e. $\pi = 0$; see \citet[Theorem 5 and its proof]{guo2015stepwise}.

\item $\STSdir$:  \algref{Storeydir} by setting $\lambda = 0.5$.

\item aSTS$_\text{dir}$: The  $\STSdir$ procedure, except with an automatic (data-driven) choice for $\lambda$ as described in \appref{auto_lambda}. It has no proven theoretical control of the $\FDRdir$.

\item ZDIRECT: \algref{zdirect}, by initializing $\cM_1 = \{i: u_i' \leq 0.2 \text{ or } u_i' \geq 0.8\}$ based on the $\{\tilde{u}_{i, 0}\}_{i=1}^m = \{\tilde{u}_i'\}_{i=1}^m$, which respects condition (C1). Thereafter, $\cM_t$ is updated as in \secref{oracle}, where the optimization in \eqref{optim} is performed using a solver in the R package \texttt{Rmosek} \citep{mosek}. The points $a_k$ are picked to give a large and dense grid; in particular, for the positive supports, the minimum and maximum are set as $a_1 = 10^{-1}$ and $a_K = 2 \sqrt{  \max_i z_i'^2 - 1}$, with the rest set as $a_{k+1} = \sqrt{2} a_k \leq a_K$ based on the multiplicative factor $\sqrt{2}$ (and so, $K$ is implicitly determined). The negative supports are set by taking $a_{-1} = -a_1, \dots, a_{-K} = -a_K$. This grid follows the recommendation of \citet{stephens2017false} except $a_K$ is determined with $z_i'$'s instead of $z_i$ to observe the masking condition (C1). Moreover, we set $\lambda_k = 0.8$ for all $k = -K, \dots, 0, \dots, K$. This further regularizes the estimation by encouraging sparsity in the estimates of the mixing proportions ${\bf w}$, and provides consistently good performance. Lastly, to speed up the algorithm,  $\hat{g}_t(\cdot)$ is only re-estimated by \eqref{optim} for every $\lceil m/200 \rceil$ steps, i.e., the same $\hat{g}_t(\cdot)$ is used $\lceil m/200 \rceil$ times to update the candidate rejection and acceptance sets before the algorithm terminates.

\end{enumerate}

\subsection{Results}
The empirical $\FDRdir$ and power of the different methods implemented for the target $\FDRdir$ level  $q = 0.1$ are evaluated with 1000 sets of repeatedly generated  $\{z_i, \theta_i\}_{i \in [m]}$. The results are shown in \figref{result}, where the power is shown as the \emph{true positive rate}, defined as 
$
\mathbb{E}\left[\frac{ \sum_{i \in \cR} {\bf 1}(\text{sgn}(\mu_i) =  \widehat{sgn}_{i} )|}{ 1 \vee |\cR|} \right] 
$ for a generic procedure $(\widehat{sgn}_{i})_{i \in \cR}$, 
which some consider to be  more illustrative than the ETD. The following observations can be made:

\begin{enumerate} [label=(\alph*)]
\item Throughout, the only methods that can visibly control $\FDRdir$ under the target $q = 0.1$ in all settings are LSFR, $\BHdir$, ZDIRECT and $\STSdir$, precisely the ones with  theoretical guarantees. However, LSFR is not implementable in practice and only has $\FDRdir$ guarantee under the Bayesian formulation in \eqref{prior}. While ZDIRECT and $\STSdir$ still lag considerably  behind in power compared to LFSR in  settings with small $w$ (or small $\pi$, as a frequentist analogue), their power advantage over $\BHdir$ becomes substantial as $w$ approaches $0$. Across the board, aSTS$_\text{dir}$, which is calibrated with a data-driven $\lambda$ and has no theoretical guarantee, displays slightly better power than $\STSdir$, but also violates the $\FDRdir$ target ever so slightly for large $\xi$ when $w = 0.8$. 

\item ASH  is the one practical method that overall matches LFSR closest in power, but
severely violates the desired $\FDRdir$ level in some settings when $w \in \{ 0.2,0.5 \}$. This is not surprising because the unimodal working model in \eqref{halfuniform}, which ASH is based on, is misspecified for many of the multimodal data generating $g(\cdot)$ in \figref{gplot}.

\item  GR visibly violates the $\FDRdir$ target when $w \in \{ 0.5, 0.8 \}$, and severely so for $w = 0.8$. When $w = 0$, the only case where GR can provably control the $\FDRdir$, GR's power is at best comparable to  $\STSdir$ and ZDIRECT when the signal size $\xi$ is small, but inferior to them when $\xi$ is large.

\item  ZDIRECT's power is considerably better than $\STSdir$ when $v = 1$, the most asymmetric setting for $g$.  As discussed in the remarks of \secref{oracle}, by attempting to mimic the operational characteristics of the ODP via estimating the $\lfsr_i$ quantities in \eqref{lfsrestimate},  ZDIRECT has the potential to leverage asymmetry in the distribution of the $z$-values to boost testing power, just like the ODP does. This is even more remarkable, considering that the working model \eqref{halfuniform} is obviously misspecified for the true $\theta$-generating prior in \eqref{simul_prior}, which attests to the practical usefulness of \cite{stephens2017false}'s unimodal assumption when combined with ZDIRECT's data masking mechanism to ensure strong $\FDRdir$ control (\thmref{zdirectcontrol}).

\end{enumerate}

% figure
\begin{figure*}[htp]
\centering
\includegraphics[scale=0.66]{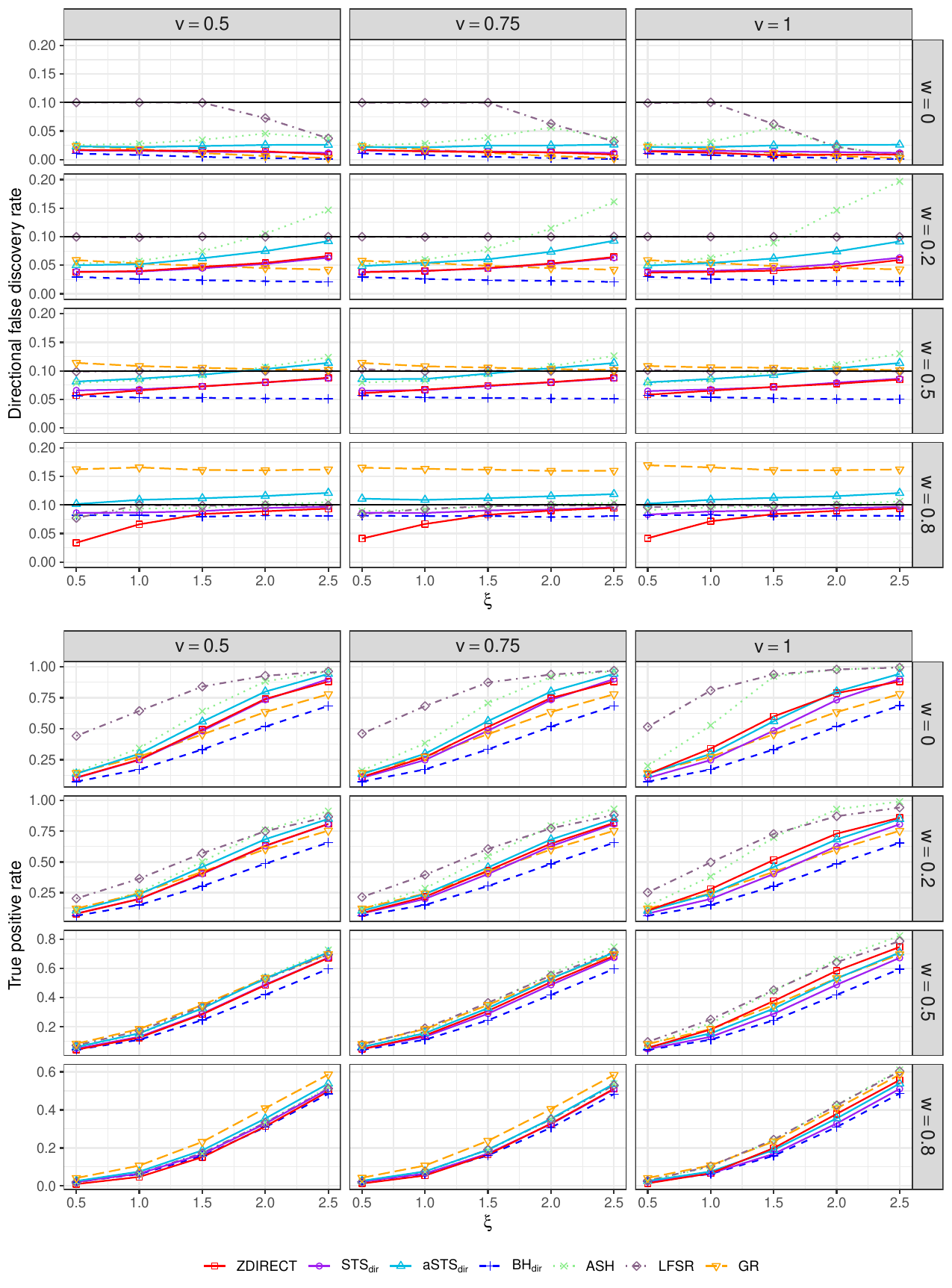}
\caption{
Empirical directional false discovery rate and true positive rates of the seven compared methods for the simulations in \secref{simul}; each method was implemented at a target $\FDRdir$ level $q =0.1$ (black horizontal lines).
% \texttt{densest zeros}, \texttt{denser zeros}, \texttt{sparse zeros} and \texttt{no zeros} correspond to $w = 0.8, 0.5, 0.2, 0$ and $0$ respectively. \texttt{symmetric}, \texttt{more asymmetric} and \texttt{most asymmetric} correspond to $v = 0.5, 0.75, 1$ respectively.
 }
\label{fig:result}
\end{figure*}

\newpage
\section{Discussion} \label{sec:conclude}
We have proved that, under independence and upon augmenting sign declarations, \cite{storey2004strong}'s adaptive procedure and ZDIRECT, a particular implementation of the recently introduced line of adaptive ``data masking'' algorithms, can offer  $\FDRdir$ control in the strong sense. These results are particularly important when the parameter configurations contain little to no true nulls because adaptive procedures precisely reap the most power benefit in such scenarios. Moreover, under ``non-sparse-signal'' settings, $\FDRdir$ is arguably a more meaningful error rate to be controlled than the $\FDR$. Both methods require tuning parameters; in our experience, the simple choice of $\lambda = 0.5$ for $\STSdir$ and $\lambda_k = 0.8$ for ZDIRECT have consistently given us competitive power performance, even for some less instructive simulation setups considered in earlier versions of this paper. For ZDIRECT, while other working models that may  further boost the power can be  deployed, we find the current implementation based on  \citet{stephens2017false}'s unimodal model to be attractive, as the  optimization under the hood is numerically very stable and fast.  While our theory doesn't cover settings where the $z$-values can be dependent, additional simulation results in this vein are included in  \appref{simResultsExtra}.

An interesting dual problem to sign declarations is to construct, for each $i$ in a data-dependent selected subset 
$\mathcal{R} \subseteq [m]$,  
a  confidence interval $CI_i \subseteq \mathbb{R}$   such that
\begin{enumerate}[label=(\alph*)]
\item each $CI_i$ is \emph{sign-determining},
 i.e. $CI_i \subseteq  (- \infty, 0]$ or $CI_i \subseteq  (0, \infty)$, and 
\item  the false coverage rate (FCR)
\[
\mathbb{E} \left[ \frac{ |\{i: \theta_i \in  CI_i\}|}{ 1 \vee |\mathcal{R}|}\right]
\]
is controlled under a desired level $q \in (0, 1)$.
\end{enumerate}
This paradigm of inference has been recently suggested in  \citet{weinstein2020selective}, and they proposed a first procedure  that constructs such selective sign-determining confidence intervals. The adjective ``selective'' indicates that the set $\cR$ is also chosen based on the same data  $\{z_i\}_{i=1}^m$  the $CI_i$'s are constructed with. 
 Note, $\STSdir$ and ZDIRECT  do correspond to such procedures that construct trivially long intervals  
\[
CI_i = 
  \begin{cases} 
 (0, \infty)& \text{if }  z_i > 0 \\
(-\infty, 0)  & \text{if }    z_i < 0
  \end{cases}
\]
for each $i$ in their rejection sets. It will be a challenging but meaningful task to devise non-trivial selective sign-determining confidence intervals where the target set $\cR$ is chosen in a more adaptive manner akin to $\STSdir$ and ZDIRECT, to offer more powerful alternative procedures to \citet[Definition 2]{weinstein2020selective}'s procedure.

%%%%%%%%%%%%%%%%%%%%%%%%%%%%%%%%%%%%%%%%%%%%%%
%% Example with multiple Appendixes:        %%
%%%%%%%%%%%%%%%%%%%%%%%%%%%%%%%%%%%%%%%%%%%%%%
\begin{appendix}
\section{Additional content for \secref{Storeydir}}

\subsection{Proof of \thmref{Storeydircontrol}}\label{app:Storeydir}

We shall first state four intermediate results, which allow us to extend the arguments in \citet[Section 4.3]{storey2004strong} to prove the $\FDRdir$ controlling properties of $\STSdir$ under 
\assumpsref{symm} and \assumpssref{mlr}.

\begin{lemma}\label{lem:mlrexploit}
    Under \assumpref{mlr}, for $z \in \mathbb{R}$, 
    \begin{enumerate}[label=(\alph*)]
        \item  if $\theta_i < 0$, then $\left[ \frac{f_{i,\theta_i} (z) }{ f_{i,0} (z) }  \right] \left[ 1 - F_{i,0}(z) \right] - \left[ 1 - F_{i,\theta_i} \left( z ) \right) \right] \geq 0$;
        \item if $\theta_i > 0$, then $\left[ \frac{f_{i,\theta_i} (z) }{ f_{i,0} (  z) }  \right] F_{i,0}(z)  -  F_{i,\theta_i} ( z )  \geq 0$.
    \end{enumerate}

\end{lemma}
\begin{proof}[Proof of \lemref{mlrexploit}]
    To prove statement (a), let $\theta_i < 0$ and let $z_0, z_1 \in \mathbb{R}$ such that $z_0 < z_1$. By \assumpref{mlr}, we have that
    \begin{align*}
         \frac{f_{i,\theta_i} (z_0) }{ f_{i,0} (z_0) } \geq \frac{f_{i,\theta_i} (z_1) }{ f_{i,0} (z_1) }.  
    \end{align*}
    Multiplying both sides by $f_{i,0}(z_1)$ and integrating over $z_1$ from $z_0$ to $\infty$ yields
    \begin{align*}
         \left[ \frac{f_{i,\theta_i} (z_0) }{ f_{i,0} (z_0) }  \right] \left[ 1 - F_{i,0}(z_0) \right] \geq \left[ 1 - F_{i,\theta_i} ( z_0 ) \right], 
    \end{align*}
    which proves statement (a). The proof of statement (b) follows analogously to that of (a). 
\end{proof}

In the lemma below, the probability operator $P_{\theta_i = 0}(\cdot)$ emphasizes the law is driven by a  value of $\theta_i$ equal to zero;  the operators $P_{\theta_i < 0}(\cdot)$  and $P_{\theta_i > 0}(\cdot)$ have similar meanings. 
\begin{lemma}\label{lem:condprobs}
    Under \assumpsref{symm} and \assumpssref{mlr}, for $0 < s \leq t \leq 1$,
    \begin{enumerate}[label=(\roman*)]
    \item $P_{\theta_i = 0} \left( p_i \leq s |p_i \leq t , z_i \neq 0    \right) = s/t$;
    \item $P_{\theta_i < 0 } \left( p_i \leq s |p_i \leq t , z_i \geq 0   \right) \leq s/t$;
    \item $P_{\theta_i > 0}  \left( p_i \leq s |p_i \leq t , z_i \leq 0
    \right) \leq s/t$.
    \end{enumerate}
\end{lemma}
\begin{proof}[Proof of \lemref{condprobs}]
For any $x \in (0, 1]$, we first rewrite
\begin{align}
 P_{\theta_i = 0} \left( p_i \leq x, z_i \neq 0   \right) &=  \left[ 1 - F_{i,0} \left( - F^{-1}_{i,0}(x/2)  \right) \right]  + F_{i,0} \left( F^{-1}_{i,0}(x/2)  \right) = x; \label{prob0} \\
    P_{\theta_i < 0}  \left( p_i \leq x, z_i \geq 0 \right) &=  1 - F_{i,\theta_i} \left( - F^{-1}_{i,0}(x/2)  \right); \label{prob-} \\
    P_{\theta_i >0} \left( p_i \leq x, z_i \leq 0  \right) & = F_{i,\theta_i} \left( F^{-1}_{i,0}(x/2)  \right). \label{prob+}
\end{align}

$(i)$ follows from \eqref{prob0}  since
$
    P_{\theta_i = 0} \left( p_i \leq s |p_i \leq t , z_i \neq 0  \right) = \frac{P_{\theta_i = 0} \left( p_i \leq s, z_i \neq 0   \right)}{P_{ \theta_i = 0} \left( p_i \leq t, z_i \neq 0   \right)} = \frac{s}{t}
$.

 $(ii)$ is obvious  when $s = t$; when $ s < t$, by applying the mean value theorem on $\frac{1 - F_{i,\theta} \left( - F^{-1}_{i,0}(x/2)  \right)}{x}$ as a function in $x$  in light of \eqref{prob-}, there exists  $c \in (s,t)$   such that, for $y \equiv F^{-1}_{i,0} (c/2)$,
 \begin{align}
    \frac{ \left[ \tfrac{P_{\theta_i <0} \left( p_i \leq t,  z_i \geq 0      \right)}{t} \right] - \left[ \tfrac{P_{\theta_i <0} \left( p_i \leq s, z_i \geq 0     \right)}{s} \right] }{t - s}  \notag
%     =  \frac{\left[ \frac{f_{i,\theta_i} ( - F^{-1}_{i,0} (c/2)) }{ 2 f_{i,0} ( F^{-1}_{i,0} (c/2)) }  \right] c - \left[ 1 - F_{i,\theta_i} \left( - F^{-1}_{i,0}(c/2)  \right) \right]  }{c^2}   
    &= \frac{\left[ \frac{f_{i,\theta_i} ( - y) }{ f_{i,0} (  y ) }  \right] F_{i,0}(y) - \left[ 1 - F_{i,\theta_i} \left( - y  \right) \right]  }{4 F_{i,0}(y)^2}  \\
    &= \frac{\left[ \frac{f_{i,\theta_i} ( - y) }{ f_{i,0} ( - y ) }  \right] \left[ 1 - F_{i,0}(- y) \right] - \left[ 1 - F_{i,\theta_i} \left( - y ) \right) \right]  }{4 F_{i,0}(y)^2}  \label{last_line_for_i},
\end{align}
where the last equality follows from the symmetry of $F_{i, 0}$ in \assumpref{symm}. Since $\theta_i < 0$, by applying \lemref{mlrexploit}$(a)$ to the numerator in \eqref{last_line_for_i}, we get that
\[
\left[ \frac{P_{\theta_i < 0} \left( p_i \leq t,  z_i \geq 0      \right)}{t} \right] - \left[ \frac{P_{\theta_i < 0} \left( p_i \leq s, z_i \geq 0      \right)}{s} \right] \geq 0,
\]
which is equivalent to $\frac{P_{\theta_i < 0 } \left( p_i \leq s, z_i \geq 0     \right)}{P_{\theta_i < 0 } \left( p_i \leq t, z_i \geq 0     \right)} \leq \frac{s}{t}
$, and $(ii)$ is proved.
%\begin{align*}
%     \frac{ \left[ \frac{P \left( p_i \leq t,  z_i \geq 0    ;\theta_i < 0  \right)}{t} \right] - \left[ \frac{P \left( p_i \leq s, z_i \geq 0    ;\theta_i < 0  \right)}{s} \right] }{t-s}\geq 0 \text{ which implies } \frac{P \left( p_i \leq s, z_i \geq 0    ;\theta_i < 0  \right)}{P \left( p_i \leq t, z_i \geq 0    ;\theta_i < 0  \right)} \leq \frac{s}{t}
%\end{align*}
%which finishes the proof of statement (ii).
 The proof for $(iii)$ is analogous to that of $(ii)$, using the mean value theorem on $\frac{F_{i,\theta_i} \left( F^{-1}_{i,0}(x/2)  \right)}{x}$, \eqref{prob+} and \lemref{mlrexploit}$(b)$.
\end{proof}

\begin{lemma}\label{lem:supermartingale} Let $\mathcal{S}(t)$ be defined as in \thmref{Storeydircontrol}.
Under \assumpsref{symm} and \assumpssref{mlr}, as well as the independence among $z_1, \dots, z_m$, $|\mathcal{S}(t)|/t$ for $0  \leq t < 1$ is a  supermartingale with time running backward, with respect to the filtration 
$$\mathcal{F}_t = \sigma \left( \left\{ \textbf{1}(p_i \leq x), \textbf{1}(p_i \leq x, \text{sgn}(z_i) \neq \text{sgn}(\theta_i)) \right\}_{i \in [m]}: 0 < t \leq x \leq 1 \right).$$
That is, for $0 < s \leq t \leq 1$,
$
    \mathbb{E}_{\boldsymbol{\theta}} \left[ \left. \frac{|\mathcal{S}(s)|}{s} \right| \mathcal{F}_t \right] \leq \frac{|\mathcal{S}(t)|}{t}.
$
\end{lemma}

\begin{proof}[Proof of \lemref{supermartingale}] Since $|\mathcal{S}(s)|$ can be written as
\begin{align}\label{Ssum}
    |\mathcal{S}(s)| &= \sum_{i: \theta_i = 0} \textbf{1}( p_i \leq s, z_i \neq 0) +  \sum_{i: \theta_i > 0} \textbf{1}( p_i \leq s, z_i \leq 0) + \sum_{i: \theta_i < 0} \textbf{1}( p_i \leq s, z_i \geq 0),
\end{align}
it is not difficult to observe from statements $(i)$-$(iii)$ of \lemref{condprobs} that $|\mathcal{S}(s)|$ given $\mathcal{F}_t$ is stochastically dominated by the $\text{Binomial}(|\mathcal{S}(t)|,s/t)$ distribution. Hence, $\mathbb{E}_{\boldsymbol{\theta}}[|\mathcal{S}(s)||\mathcal{F}_t] \leq |\mathcal{S}(t)| \cdot (s/t)$.
\end{proof}

\begin{lemma}\label{lem:exresult}
If $Y \sim \text{Binomial}(n,\lambda)$, then for any $\lambda \in (0,1)$ and $n \in \mathbb{N}$,
    \[
    D_n(\lambda) \equiv \mathbb{E} \left[ \frac{Y}{n - Y + 1} \right] = \frac{(1 - \lambda^n) \lambda}{1 - \lambda},
    \]
where $D_n(\lambda)$ is strictly increasing in $\lambda$.
\end{lemma}

\begin{proof}[Proof of \lemref{exresult}]
By direct computation, 
\begin{multline*}
   D_n(\lambda) = \sum_{i = 1}^n {n \choose i} (1 - \lambda)^{n-i} \lambda^i \cdot \frac{i}{n-i+1} \\
   = \frac{\lambda}{1 - \lambda} \sum_{i = 1}^n {n \choose i-1} (1 - \lambda)^{n-i+1} \lambda^{i-1} = \frac{\lambda}{1 - \lambda} \cdot (1 - \lambda^n)
\end{multline*}
which proves the expectation result. Taking the derivative of $D_n(\lambda)$ with respect to $\lambda$ yields
\begin{align*}
    D'_n(\lambda) = \frac{\lambda^n (n \lambda - n - 1) + 1}{(\lambda - 1)^2}.
\end{align*}
Let the numerator of $D'_n(\lambda)$ be denoted as $\dot{D}'_n(\lambda) \equiv \lambda^n (n \lambda - n - 1) + 1$.
To prove that $D_n(\lambda)$ is strictly increasing in $\lambda$, we will show that $\dot{D}'_{n}(\lambda) > 0$ by induction. Suppose $\dot{D}'_{N}(\lambda) > 0$ is true for \textit{some} fixed $N \in \mathbb{N}$. Then
\begin{align*}
    \dot{D}'_{N+1}(\lambda) &= \lambda^{N+1} ((N+1)\lambda-(N+1)-1) + 1 \\
    &= \lambda (\dot{D}'_N(\lambda) + \lambda^{N} ( \lambda - 1) - 1) + 1 \\
    &> \lambda ( \lambda^N (\lambda - 1) - 1) + 1\\
    &= (1 - \lambda) ( 1- \lambda^{N+1})  \\
    &> 0
\end{align*}
where the first inequality stems from the inductive condition $\dot{D}'_{N}(\lambda) > 0$. Since $\dot{D}'_{1}(\lambda) = (\lambda-1)^2 > 0$, it follows  that $\dot{D}'_{n}(\lambda) > 0$ for any $\lambda \in (0,1)$ and $n \in \mathbb{N}$.
\end{proof}

Now we can begin the proof of \thmref{Storeydircontrol}. First, we can write
\begin{align*}
    \mathbb{E}_{\boldsymbol{\theta}} \left[ \frac{|\mathcal{S}(t^{\lambda}_{q})|}{ |\mathcal{R}(t^{\lambda}_{q})| \vee 1  } \right]= \mathbb{E}_{\boldsymbol{\theta}} \left[ \frac{|\mathcal{S}(t^{\lambda}_{q})|}{ |\mathcal{R}(t^{\lambda}_{q})| \vee 1  } ; \widehat{\text{FDR}}_{\lambda}(\lambda) \geq q \right] + \mathbb{E}_{\boldsymbol{\theta}} \left[ \frac{|\mathcal{S}(t^{\lambda}_{q})|}{ |\mathcal{R}(t^{\lambda}_{q})| \vee 1  } ; \widehat{\text{FDR}}_{\lambda}(\lambda) < q \right].
\end{align*}
If $\widehat{\text{FDR}}_{\lambda}(\lambda) \geq q$, then $t^{\lambda}_q \leq \lambda$. Hence, $|\mathcal{R}(t^{\lambda}_{q})| \vee 1  \geq \hat{\pi}_0(\lambda) t^{\lambda}_{q} m / q$ and so
\begin{align*}
    \mathbb{E}_{\boldsymbol{\theta}} \left[ \frac{|\mathcal{S}(t^{\lambda}_{q})|}{ |\mathcal{R}(t^{\lambda}_{q})| \vee 1  } ; \widehat{\text{FDR}}_{\lambda}(\lambda) \geq q \right] &\leq \mathbb{E}_{\boldsymbol{\theta}} \left[ q \frac{1 - \lambda}{|\{i:p_i > \lambda\}| + 1}  \frac{|\mathcal{S}(t^{\lambda}_{q})|}{ t^{\lambda}_{q} } ; \widehat{\text{FDR}}_{\lambda}(\lambda) \geq q \right] \\
    &= \mathbb{E}_{\boldsymbol{\theta}} \left[ q \frac{1 - \lambda}{|\{i:p_i > \lambda \}| + 1}  \mathbb{E}_{\boldsymbol{\theta}} \left[ \left. \frac{|\mathcal{S}(t^{\lambda}_{q})|}{ t^{\lambda}_{q} }  \right| \mathcal{F}_{\lambda} \right] ; \widehat{\text{FDR}}_{\lambda}(\lambda) \geq q \right] \\
    &\leq \mathbb{E}_{\boldsymbol{\theta}} \left[ q \frac{1 - \lambda}{|\{i:p_i > \lambda \}| + 1}  \frac{|\mathcal{S}(\lambda)|}{ \lambda }  ; \widehat{\text{FDR}}_{\lambda}(\lambda) \geq q \right]
\end{align*}
where the last step follows by \lemref{supermartingale} and the optional stopping  theorem since $t^{\lambda}_{q}$ is a stopping time with respect to $\mathcal{F}_t$ with time running backward. If $\widehat{\text{FDR}}_{\lambda}(\lambda) < q$, then $t^{\lambda}_q = \lambda$ and so 
\begin{equation*}
    \mathbb{E}_{\boldsymbol{\theta}} \left[ \frac{|\mathcal{S}(t^{\lambda}_{q})|}{ |\mathcal{R}(t^{\lambda}_{q})| \vee 1  } ; \widehat{\text{FDR}}_{\lambda}(\lambda) < q \right] \leq \mathbb{E}_{\boldsymbol{\theta}} \left[ q \frac{1 - \lambda}{|\{i:p_i > \lambda\}| + 1}  \frac{|\mathcal{S}(\lambda)|}{ \lambda }  ; \widehat{\text{FDR}}_{\lambda}(\lambda) < q \right].
\end{equation*}
Hence,
\begin{equation*}
    \mathbb{E}_{\boldsymbol{\theta}} \left[ \frac{|\mathcal{S}(t^{\lambda}_{q})|}{ |\mathcal{R}(t^{\lambda}_{q})| \vee 1  } \right]\leq \mathbb{E}_{\boldsymbol{\theta}} \left[ q \frac{1 - \lambda}{|\{i:p_i > \lambda\}| + 1}  \frac{|\mathcal{S}(\lambda)|}{ \lambda }  \right] 
     \leq \mathbb{E}_{\boldsymbol{\theta}} \left[ q \frac{ |\mathcal{S}(\lambda)|}{|\mathcal{S}(1)| - |\mathcal{S}(\lambda)| + 1}  \frac{ 1 - \lambda }{ \lambda }  \right].
\end{equation*}
By taking $s = \lambda$ and $t=1$ in statements $(i)$-$(iii)$ of \lemref{condprobs}, in light of the equality in \eqref{Ssum}, it is not difficult to see that $|\mathcal{S}(\lambda)|$ given $|\mathcal{S}(1)|$ is stochastically dominated by the $\text{Binomial}(|\mathcal{S}(1)|, \lambda)$ distribution. Hence,
\begin{equation} \label{almost_there}
\mathbb{E}_{\boldsymbol{\theta}} \left[ q \frac{ |\mathcal{S}(\lambda)|}{|\mathcal{S}(1)| - |\mathcal{S}(\lambda)| + 1}  \frac{ 1 - \lambda }{ \lambda }  \right] \leq   \Big(1 - \mathbb{E}_{\boldsymbol{\theta}}[\lambda^{|\mathcal{S}(1)|}]\Big) q
    \end{equation}
    by \lemref{exresult}. Combining \eqref{almost_there} with
  $
    1 -  \mathbb{E}_{\boldsymbol{\theta}}[\lambda^{|\mathcal{S}(1)|}] \leq  1 - \lambda^{\mathbb{E}_{\boldsymbol{\theta}}[|\mathcal{S}(1)|]} \leq 1
     $, a consequence of  Jensen's inequality, \thmref{Storeydircontrol} is proved.

\subsection{Automatic $\lambda$ selection procedure}\label{app:auto_lambda}

Two inputs are required for this procedure: $B$, the number of bootstrap samples; and $\Lambda$, a set of candidate values for $\lambda$. Our recommendations are $B = 1000$ and $\Lambda = \{ 0.05, 0.10, \dots, 0.95 \}$. The procedure is summarized in the following algorithm.

\begin{enumerate}[label=(\arabic*)]
    \item Compute $\hat{\pi}(\lambda')$ for each $\lambda' \in \Lambda$.
    \item For each $\lambda' \in \Lambda$, construct $B$ bootstrap estimates $\{ \hat{\pi}^{b}(\lambda') \}^B_{b=1}$ by bootstrap sampling the $p$-values.
    \item Compute $\hat{\pi}^{\text{min}} \equiv \min_{\lambda' \in \Lambda}\ \hat{\pi}(\lambda')$.
    \item For each $\lambda' \in \Lambda$, compute
    \begin{align*} \label{mse} 
        \widehat{MSE}(\lambda') = \frac{1}{B} \sum^B_{b=1} [\hat{\pi}^{b} (\lambda') - \hat{\pi}^{\text{min}} ]^2.
    \end{align*}
  \item Output the estimated optimal tuning parameter $\hat{\lambda} = \text{argmin}_{\lambda' \in \Lambda}\{  \widehat{MSE}(\lambda') \}$.
\end{enumerate}

 The above algorithm is nearly identical to that of \cite{storey2004strong}'s automatic $\lambda$ selection algorithm (Section 6), except that \cite{storey2004strong} omit the additive factor  ``1'' in the numerator of all the estimators for $\pi$ involved, but we retain it to
% we used $\hat{\pi}^{\text{min}} \equiv \min_{\lambda' \in \Lambda}  \frac{1+|\{i:p_i > \lambda'\}| }{(1 - \lambda')m} $ instead of \cite{storey2004strong}'s $\min_{\lambda' \in \Lambda}  \frac{|\{i:p_i > \lambda'\}| }{(1 - \lambda')m} $ in Step 3, to
  robustify the $\FDRdir$ controlling property of the resulting procedure. Regardless, the intuition behind is  the same, i.e. choose a $\lambda$ which minimizes an estimated mean square error.

% new appendix
\section{Additional content for \secref{zdirect}}

\subsection{Computation of the component loglikelihoods} \label{app:compLike}

We discuss computations of the likelihoods in \eqref{component_log_lik} when $f_{i, \theta}$ belongs to the normal family $N(\theta, \sigma_i^2)$ or the noncentral $t$-distributional family $NCT(\theta, \nu_i)$.

\begin{enumerate} [label=(\roman*)]
\item
$N(\theta, \sigma_i)$: In this case, each  $ l_{k, i, t}$ in  \eqref{component_log_lik} has the explicit analytic form
\[
 l_{k, i, t} = 
  \begin{cases} 
   \frac{ \Phi(z_i /\sigma_i ) - \Phi((z_i - a_k)/\sigma_i ) }{a_k}  +   \frac{ \Phi(\widecheck{z}_i/\sigma_i  ) - \Phi( (\widecheck{z}_i - a_k)/\sigma_i ) }{a_k} & \text{if }  i \in \mathcal{M}_t  \text{ and } k \geq 1\\
    \frac{ \Phi(z_i /\sigma_i ) - \Phi( (z_i - a_k)/\sigma_i ) }{a_k}  & \text{if }    i \in [m] \backslash \mathcal{M}_t  \text{ and } k \geq 1\\
     \frac{ \Phi((z_i  - a_k)/\sigma_i ) - \Phi(z_i /\sigma_i ) }{- a_k} +  \frac{ (\Phi(\widecheck{z}_i - a_k)/\sigma_i ) - \Phi(\widecheck{z}_i /\sigma_i ) }{- a_k} & \text{if }  i \in \mathcal{M}_t  \text{ and } k \leq -  1\\
   \frac{ \Phi( (z_i  - a_k)/\sigma_i ) - \Phi(z_i /\sigma_i ) }{- a_k}  & \text{if }    i \in [m] \backslash \mathcal{M}_t  \text{ and } k \leq -1
  \end{cases}.
\]

\item $NCT(\theta, \nu_i)$: Without sophisticated numerical integration methods, it may be hard to obtain good numerical values for the quantities in \eqref{component_log_lik}. However,  approximation methods can be potentially leveraged; in what follows we assume the  common use case where $\nu_i = \nu$ for all $i \in [m]$. In a variance-stabilizing manner,  \citet[Section 2]{laubscher1960normalizing} suggests that, if $z_i$ is a noncentral $t$-distributed random variable with noncentrality parameter $\theta$ and degree $\nu \geq 4$, by \emph{bijectively} transforming $z_i$ to the variable 
\[
\xi_i \equiv \alpha \sinh^{-1} (\beta z_i),
\]
where $\alpha = \alpha(\nu)$ and $\beta = \beta(\nu)$ are positive numbers depending only on $\nu$, 
\[
\text{$\xi_i$ is approximately distributed as $N(\gamma , 1)$}
\] 
for the mean
\[
\gamma \equiv \alpha \sinh^{-1} \Bigg( \beta \cdot \theta  \cdot \frac{  \Gamma(\nu/2- 1/2) \sqrt{\nu/2} }{\Gamma (\nu/2)}  \Bigg),
\]
which is also strictly increasing in $\theta$.
 Hence, by also letting $\widecheck{\xi_i} \equiv \alpha \sinh^{-1} (\beta \widecheck{z}_i)$, one can alternatively implement ZDIRECT by replacing the component likelihoods in \eqref{component_log_lik} with
\[
 l_{k, i, t}' = 
  \begin{cases} 
   \frac{ \Phi(\xi_i ) - \Phi(\xi_i - a_k) }{a_k}  +   \frac{ \Phi(\widecheck{\xi}_i ) - \Phi(\widecheck{\xi}_i - a_k) }{a_k} & \text{if }  i \in \mathcal{M}_t  \text{ and } k \geq 1\\
    \frac{ \Phi(\xi_i ) - \Phi(\xi_i - a_k) }{a_k}  & \text{if }    i \in [m] \backslash \mathcal{M}_t  \text{ and } k \geq 1\\
     \frac{ \Phi(\xi_i  - a_k) - \Phi(\xi_i ) }{- a_k} +  \frac{ \Phi(\widecheck{\xi}_i - a_k) - \Phi(\widecheck{\xi}_i ) }{- a_k} & \text{if }  i \in \mathcal{M}_t  \text{ and } k \leq -  1\\
   \frac{ \Phi(\xi_i  - a_k) - \Phi(\xi_i ) }{- a_k}  & \text{if }    i \in [m] \backslash \mathcal{M}_t  \text{ and } k \leq -1
  \end{cases}.
\]
In doing so, we have essentially imposed the prior  $g(\cdot, {\bf w})$ in \eqref{halfuniform} on $\gamma$ instead of $\theta$, but it doesn't change things in the grand scheme as it still approximates a unimodal density about zero on $\theta$; note that $\sinh^{-1}(0) = 0$. Importantly, we are still working with the partial data $\{\tilde{z}_{i, t}\}_{i \in [m]}$ so strong $\FDRdir$ control is guaranteed by virtue of \thmref{zdirectcontrol}. Other such  strategies may also be explored, possibly based on ideas from \citet{kraemer1979central} and other references therein. 

\end{enumerate}

\subsection{Proof of \thmref{zdirectcontrol}} \label{app:FDRdirpf}

We first  quote \citet[Lemma 2]{lei2018adapt}, a  fundamental tool for developing data-masking algorithms. 

\begin{lemma} \label{lem:leiResult}
  Suppose that, conditionally on the $\sigma$-field $\mathcal{G}_{-1}$, $b_1,\ldots,b_n$ are independent Bernoulli random variables with $P(b_i = 1 \mid \mathcal{G}_{-1}) = \rho_i \geq \rho > 0$, almost surely. Also suppose that $[n] \supseteq \mathcal{C}_0 \supseteq \mathcal{C}_1 \supseteq \cdots$, with each subset $\mathcal{C}_{t+1}$ measurable with respect to
  \[
 \mathcal{G}_t = \sigma\left\{\mathcal{G}_{-1}, \mathcal{C}_t, (b_i)_{i \notin \mathcal{C}_t}, \sum_{i \in \mathcal{C}_t} b_i\right\}.
  \]

  If $\hat{t}$ is an almost-surely finite stopping time with respect to the filtration $(\mathcal{G}_t)_{t \geq 0}$, then
  \[
 \mathbb{E}\left[\frac{1 + |\mathcal{C}_{\hat{t}}|}{1 + \sum_{i\in \mathcal{C}_{\hat{t}}} b_i} \mid  \mathcal{G}_{-1}\right]  \leq \rho^{-1}.
  \]
\end{lemma}

In  \lemref{leiResult}, we remark that $(\mathcal{G}_t)_{t \geq 0}$ defines a filtration precisely because $\mathcal{C}_{t+1} \in \mathcal{G}_t$. 
\begin{proof}[Proof of \thmref{zdirectcontrol}]
 The arguments below are inspired by those from \citet{barber2019knockoff} for establishing the $\FDRdir$ controlling property of the knockoff filter for variable selection in linear regressions.
To begin our proof, write the directional false discovery proportion as
\[
\text{FDP}_{\text{dir}}(\hat{t}) = \frac{|\{i : \text{sgn}(\theta_i) \neq  \text{sgn}(z_i)    \text{ and }  i \in \mathcal{R}_{\hat{t}} \}|}{ 1 \vee |\mathcal{R}_{\hat{t}}|} =   \frac{|\{i : \text{sgn}(\theta_i) \neq  \text{sgn}(z_i)    \text{ and }  i \in \mathcal{R}_{\hat{t}} \}|}{ 1+ |\mathcal{A}_{\hat{t}}|}  \frac{ 1+    |\mathcal{A}_{\hat{t}}|  }{ 1 \vee |\mathcal{R}_{\hat{t}}|}.
\]
Since  $\FDRdirhat (\hat{t})  =  \frac{ 1+  |\mathcal{A}_{\hat{t}}|  }{  1 \vee |\mathcal{R}_{\hat{t}}|} \leq q$ by definition,
% we have 
%$
%\mathbb{E}_{{\boldsymbol \theta}}[dFDP(\hat{t})   ] \leq q  \mathbb{E}_{{\boldsymbol \theta}}\left[ \frac{|\{i : \text{sgn}(\theta_i) \neq  \text{sgn}(z_i)    \text{ and }  i \in \mathcal{R}_{\hat{t}} \}|}{ 1+ |\mathcal{A}_{\hat{t}}|}   \right],
%$
%and hence 
it suffices to show that 
\begin{equation} \label{suffice}
 \mathbb{E}_{{\boldsymbol \theta}}\left[ \frac{|\{i : \text{sgn}(\theta_i) \neq  \text{sgn}(z_i)    \text{ and }  i \in \mathcal{R}_{\hat{t}} \}|}{ 1+  |\mathcal{A}_{\hat{t}}|}  \right] \leq 1.
\end{equation}

For each $i \in [m]$, we  define the variables
\[
b_i \equiv {\bf 1}\Big(u_i \in  (0.25, 0.75) \Big)
\]
and
\[
E_i \equiv 1 - 2b_i = \begin{cases} 
   +1   & \text{if } |z_i| \geq  |\widecheck{z}_i|\\
-1& \text{if } |\widecheck{z}_i| > | z_i|
  \end{cases},
\]
where, by the symmetry of $f_{i, 0}(\cdot)$ from \assumpref{symm}, the latter  is equal to $+1$ if $u_i$ is at least as close  as $\widecheck{u}_i$  to the  endpoints of the unit interval $[0,1]$, or equal to $-1$ otherwise. In particular, since 
\[
\cR_{\hat{t}} \equiv \{i: i \in \cM_{\hat{t}} \text{ and } u_i\in (0, 0.25]  \cup  [0.75, 1) \},
\]
 any $i \in \cR_{\hat{t}} $ must have its corresponding $u_i$ at least as close  to the two endpoints of $[0,1]$ as its reflection $\widecheck{u}_i \in [0.25, 0.75]$. As such, 
 it must always be that 
 \begin{equation} \label{rej_set_must_live_in}
 \mathcal{R}_{\hat{t}} \subseteq \{i: E_i = +1\},
 \end{equation} i.e., an element $i$ can possibly be a discovery only if $E_i = +1$. Moreover, 
define 
\[
S_i  = \text{sgn}(z_i) E_i,
\]
%\[
%S_i= \begin{cases} 
%    \text{sgn}(z_i) E_i   & \text{if } z_i \neq 0\\
%-1& \text{if } z_i = 0
%  \end{cases},
%\]
which will take on the same sign as $z_i$ if $E_i = +1$ and $z_i \neq 0$, as well as the set
%$u_i$ is at least as close as $\widecheck{u}_i$ to the endpoints of the unit interval  and let ${\bf S} = (S_1, \dots, S_m)$. We will define the set 
\[
\hat{\mathcal{H}}_0 \equiv  \{i: S_i \neq \text{sgn} (\theta_i)\}.
\]
Here, $\hat{\mathcal{H}}_0$ can act like a ``random null set'' since a false discovery precisely amounts to declaring a non-zero sign for any $i \in \hat{\mathcal{H}}_0 \cap \{i: E_i = +1\}$, in light of \eqref{rej_set_must_live_in} being always true. Hence, 
\begin{align*}
\mathbb{E}_{{\boldsymbol \theta}}\left[ \frac{|\{i : \text{sgn}(\theta_i) \neq  \text{sgn}(z_i)    \text{ and }  i \in \mathcal{R}_{\hat{t}} \}|}{ 1+  |\mathcal{A}_{\hat{t}}|} \right] 
&=
\mathbb{E}_{{\boldsymbol \theta}}\left[ \frac{| \mathcal{R}_{\hat{t}} \cap \hat{\mathcal{H}}_0 |}{ 1+  |\mathcal{A}_{\hat{t}}|} \right] \leq 
\mathbb{E}_{{\boldsymbol \theta}}\left[ \frac{| \mathcal{R}_{\hat{t}} \cap \hat{\mathcal{H}}_0 |}
{ 1+  | \mathcal{A}_{\hat{t}}\cap \hat{\mathcal{H}}_0|}  \right].
\end{align*}
The last inequality in the preceding display implies that \eqref{suffice} can be proved if it is true  that 
$
\mathbb{E}_{{\boldsymbol \theta}}\left[ \frac{|  \mathcal{R}_{\hat{t}} \cap \hat{\mathcal{H}}_0 |}
{ 1+  |\mathcal{A}_{\hat{t}}\cap \hat{\mathcal{H}}_0|} \right] \leq 1,
$
which, in turn, we will prove by showing
\begin{equation} \label{condH0}
\mathbb{E}_{{\boldsymbol \theta}}\left[ \frac{|\mathcal{R}_{\hat{t}}  \cap  \hat{\mathcal{H}}_0   |}
{ 1+  |\mathcal{A}_{\hat{t}}\cap \hat{\mathcal{H}}_0|}\bigg| \sigma \{\hat{\mathcal{H}}_0\} \right] \leq 1.
\end{equation}
%Before finishing proving \eqref{condH0}, we emphasize that knowledge of the set $ \hat{\mathcal{H}}_0$ does \emph{not} imply knowing the parameters $\theta_1, \dots, \theta_m$, even though $\hat{\mathcal{H}}_0$ is defined with them. 
%Moreover, the expectation operator $\mathbb{E}_{\boldsymbol \theta}[\cdot]$ simply suggests that the  $\theta_1, \dots, \theta_m$ are considered fixed (but still unknown), as opposed to the random parameter formulation in \eqref{gdist}. 

The rest of the proof proceeds by  setting the scene to apply \lemref{leiResult}. First,  recall $\cM_0 = [m]$ 
and let 
$
\mathcal{G}_{-1} \equiv  \sigma \{\hat{\mathcal{H}}_0  \}$.
%, where 
%\begin{equation} \label{z_prime_def}
%u' \equiv  F_{i,0}(z_i') = 
%  \begin{cases} 
%u_i\wedge \widecheck{u}_i    & \text{if } u_i , \widecheck{u}_i  \in (0, 0.5]\\
%  u_i\vee \widecheck{u}_i    & \text{if } u_i , \widecheck{u}_i  \in (0.5, 1)
%  \end{cases},
%\end{equation}
%for $z'_i$ defined in \eqref{z_prime_def},  is the  value between $u_i$ and $\widecheck{u}_i$ closer to the two ends of the unit interval. 
For $t = 0, 1, \dots$,
define 
\begin{equation*}\label{Ct_def}
\mathcal{C}_t
% = \{i: i \in \hat{\mathcal{H}}_0 \text{ and }   u_i \in (0,  s_{i, l,  t})  \cup   (1 - s_{i, r, t}, 1) \cup  (0.5 - s_{i, l, t}, 1.5 - s_{i, r, t}) \}
  \equiv \hat{\mathcal{H}}_0 \cap \mathcal{M}_t,
\end{equation*}
and the filtrations
\[
\mathcal{G}_t \equiv \sigma
 \left\{
\mathcal{G}_{-1},   \mathcal{C}_t, (b_i)_{i \not\in \mathcal{C}_t}, \sum_{i \in \mathcal{C}_t} b_i
\right\};
\]
note that $\mathcal{C}_{t+1} \in \mathcal{G}_t$ since $\hat{\mathcal{H}}_0 \in \mathcal{G}_{-1}$ and $\cM_{t+1} \in \mathcal{G}_t$ by respecting the condition (C1). 
%$\mathcal{G}_{-1} \subset \mathcal{G}_{0} \subset \mathcal{G}_{1} \cdots$ because $\mathcal{C}_{t+1}$ is measurable w.r.t. $\mathcal{G}_t$ since $\mathcal{M}_{t+1}$ is so. 
By the definitions of $\cA_t$ and $\cR_t$,  we must have that
\begin{equation} \label{cardinality_representation}
|\mathcal{A}_t  \cap \hat{\mathcal{H}}_0| = \sum_{i \in \mathcal{C}_t} b_i \text{ and } |\mathcal{R}_t \cap \hat{\mathcal{H}}_0| =  |\mathcal{C}_t| - \sum_{i \in \mathcal{C}_t} b_i. 
\end{equation}
Writing
\begin{align*}
|\mathcal{A}_t| &= |\mathcal{A}_t  \cap \hat{\mathcal{H}}_0| + | \{i:  i \not \in  \hat{\mathcal{H}}_0, i \in \mathcal{M}_t \text{ and } u_i \in  (0.25, 0.75)\}| \\
&=  |\mathcal{A}_t  \cap \hat{\mathcal{H}}_0| + | \{i:  i  \in ([m] \backslash \hat{\mathcal{H}}_0) \cap \mathcal{M}_t \text{ and } b_i = 1\}|
\end{align*}
and 
\begin{align*}
|\mathcal{R}_t| &= |\mathcal{R}_t  \cap \hat{\mathcal{H}}_0| + | \{i:  i \not \in  \hat{\mathcal{H}}_0, i \in \mathcal{M}_t   \text{ and }  u_i \in (0,  0.25]  \cup   [0.75, 1)\}| \\
&=  |\mathcal{R}_t  \cap \hat{\mathcal{H}}_0| + | \{i:   i  \in ([m] \backslash \hat{\mathcal{H}}_0) \cap \mathcal{M}_t \text{ and } b_i = 0\}|,
\end{align*}
one can see that $|\mathcal{A}_t|, |\mathcal{R}_t| \in \mathcal{G}_t$ for two reasons: First, $|\mathcal{A}_t  \cap \hat{\mathcal{H}}_0|$ and $|\mathcal{R}_t  \cap \hat{\mathcal{H}}_0|$ belong to $\mathcal{G}_t$ because of \eqref{cardinality_representation}. Second, for any   $i \in   ([m] \backslash \hat{\mathcal{H}}_0) \cap  \cM_t$, it must also be true that $i \not \in \mathcal{C}_t$ (by the definition of $\mathcal{C}_t$), which implies that $b_i$ belongs to $\mathcal{G}_t$; as such, both $| \{i:  i  \in ([m] \backslash \hat{\mathcal{H}}_0) \cap \mathcal{M}_t \text{ and } b_i = 1\}|$ and $| \{i:   i  \in ([m] \backslash \hat{\mathcal{H}}_0) \cap \mathcal{M}_t \text{ and } b_i = 0\}|$ are measurable with respect to $\mathcal{G}_t$. Hence, $\hat{t}$ is a stopping time with respect to $(\mathcal{G}_t )_{t \geq 0}$,  and is almost surely finite because ZDIRECT guarantees to terminate in light of the condition (C2). 

Lastly, by writing
\[
\mathbb{E}_{{\boldsymbol \theta}}\left[ \frac{|  \mathcal{R}_{\hat{t}} \cap  \hat{\mathcal{H}}_0 |}
{ 1+  |\mathcal{A}_{\hat{t}}\cap \hat{\mathcal{H}}_0|}\bigg|\mathcal{G}_{-1} \right] 
= \mathbb{E}_{{\boldsymbol \theta}}\left[ \frac{  |\mathcal{C}_t| -\sum_{i \in \mathcal{C}_t} b_i}
{ 1+ \sum_{i \in \mathcal{C}_t} b_i   }\bigg| \mathcal{G}_{-1} \right] 
= \mathbb{E}_{{\boldsymbol \theta}}\left[ \frac{ 1+ |\mathcal{C}_t|}
{ 1+ \sum_{i \in\mathcal{C}_t} b_i   } \bigg| \mathcal{G}_{-1} \right] -1, 
\] 
in light of \lemref{leiResult}, 
one only need to show that 
\[
P(b_i = 1 \mid \mathcal{G}_{-1})  \geq 0.5 \text{ for each }  i \in \hat{\mathcal{H}}_0 
\]
to wrap up the proof of \eqref{condH0}. We can break this into three cases; in what follows we also use the operator symbols $P_{\theta_i = 0}(\cdot)$, $P_{\theta_i > 0}(\cdot)$ and $P_{\theta_i < 0}(\cdot)$ defined immediately before \lemref{condprobs} to emphasize the underlying value of $\theta_i$ driving the law:
\begin{itemize}
\item Case 1, $\theta_i >0$:
%if $i \in \mathcal{H}_0$ and $u_i' < 0.5$:  In this case, one can write
%\[
%P(b_i = 1 \mid \mathcal{G}_{-1}) = \frac{f_{i, \theta_i}(F^{-1}_{i, 0}  (0.5 - u_i'))}{f_{i, \theta_i}(F^{-1}_{i, 0} (u_i')) + f_{i, \theta_i}(F^{-1}_{i, 0} (0.5 - u_i'))}.
%\]
%Since $ - \infty < F^{-1}_{i, 0} (u_i') < F^{-1}_{i, 0} (0.5 - u_i') < 0$
Since $i \in \hat{\mathcal{H}}_0$, under $\theta_i >0$ it must be that $S_i = -1$ or $0$. This can be true with either  $u_i \in [0.5, 0.75)$ or $u_i \in (0, 0.25 ]$, only the former of which can give $b_i = 1$. 
These two events respectively have probabilities 
\begin{multline*}
P_{\theta_i > 0} (u_i \in [0.5, 0.75) )= \int_{[F^{-1}_{i, 0} (0.5), F^{-1}_{i, 0} (0.75))} f_{i, \theta_i} (z) d z\\
 =   \int_{[F^{-1}_{i, 0} (0.5), F^{-1}_{i, 0} (0.75))} \frac{ f_{i, \theta_i} (z)}{f_{i, 0} (z)} f_{i, 0} (z) d z 
\end{multline*}
and
\begin{multline*}
P_{\theta_i > 0}(u_i \in (0, 0.25 ]) = \int_{( - \infty, F^{-1}_{i, 0} (0.25)]} f_{i, \theta_i} (z) d z \\
=   \int_{( - \infty, F^{-1}_{i, 0} (0.25)]}  \frac{ f_{i, \theta_i} (z)}{f_{i, 0} (z)} f_{i, 0} (z)  dz.
\end{multline*}
By the MLR property in   \assumpref{mlr}, $P_{\theta_i > 0}(u_i \in [0.5, 0.75)) \geq P_{\theta_i > 0}(u_i \in (0, 0.25])$ and hence
\[
P_{\theta_i > 0}(b_i = 1 \mid \mathcal{G}_{-1})  = \frac{P_{\theta_i > 0}(u_i \in [0.5, 0.75))} {P_{\theta_i > 0}(u_i \in [0.5, 0.75))+  P_{\theta_i > 0}(u_i \in (0, 0.25])} \geq 0.5.
\]
% it must be the case that $P(b_i = 1 \mid \mathcal{G}_{-1})  \geq 0.5$. 
\item Case 2, $\theta_i <  0$: The derivations are completely analogous to that of Case 1.

\item Case 3,  $\theta_i = 0$: In that case, $S_i$ can be $+1$ or $-1$. Since $u_i$ is uniformly distributed under $\theta_i = 0$, it is easy to see that
$
P_{\theta_i = 0}( b_i = 1 \mid \mathcal{G}_{-1} ) = 0.5
$
\end{itemize}
(We remark that the arguments above work precisely because $\mathcal{G}_{-1}$ only provides the meager knowledge of $\hat{\mathcal{H}}_0$, without any other knowledge about the data $\{z_i\}_{i=1}^m$.)

\end{proof}

\subsection{Proof of \thmref{ODP}} \label{app:lfsrOMTpf}

For two procedures $(\widehat{sgn}_{i}^{(1)})_{i \in \cR^{(1)}}$ and $(\widehat{sgn}_{i}^{(2)})_{i \in \cR^{(2)}}$, $(\widehat{sgn}_{i}^{(2)})_{i \in \cR^{(2)}}$  is said to improve upon $(\widehat{sgn}_{i}^{(1)})_{i \in \cR^{(1)}}$ if 
$\ETD [(\widehat{sgn}_{i}^{(2)})_{i \in \cR^{(2)}}] \geq \ETD[(\widehat{sgn}_{i}^{(1)})_{i \in \cR^{(1)}}]$ 
and 
$\FDRdir [ (\widehat{sgn}_{i}^{(2)})_{i \in \cR^{(2)}}] \leq \FDRdir [ (\widehat{sgn}_{i}^{(1)})_{i \in \cR^{(1)}}]$.

 \ \

Let ${\bf z}= (z_1, \dots, z_m)$, and let $ (\widehat{sgn}_{i})_{i \in \cR}$ be a certain procedure with $\FDRdir[(\widehat{sgn}_{i})_{i \in \cR}] \leq q$. We also write $\mathcal{R} = \mathcal{R}({\bf z})$ and $\widehat{sgn}_{i} = \widehat{sgn}_{i} ({\bf z})$ to emphasize that both are functions in ${\bf z}$. It suffices to show  that the two statements below are true:

%
%it is possible to construct an improved procedure $\{\cR', (\widehat{sgn}_{i}')_{i \in \cR'}\Bigr\}$ still with $ETP \Bigl\{\cR', (\widehat{sgn}_{i})_{i \in \cR'}\Bigr\} \geq ETP \Bigl\{\cR, (\widehat{sgn}_{i})_{i \in \cR}\Bigr\}$ and $\FDRdir \Bigl\{\cR', (\widehat{sgn}_{i})_{i \in \cR'}\Bigr\} \leq q$, if  either one or both of the followings are true:
\begin{itemize}
\item  Statement 1: If there exists  $j \in [m]$   such that one or both of the disjoint events
\[
% simulate displayed equation
\mathcal{Z}_j^{(1)}\equiv\left\lbrace {\bf z}\;\middle|\;
  \begin{varwidth}{\linewidth}
     $j \in \mathcal{R} ({\bf z})$;\\
    $P(\theta_j \leq 0 |z_j) < P(\theta_j \geq 0 |z_j)$;\\
    $ \widehat{sgn}_j ({\bf z})= -1$.
   \end{varwidth}
  \right\rbrace
  and \ \
  % simulate displayed equation
\mathcal{Z}_j^{(2)}\equiv\left\lbrace {\bf z}\;\middle|\;
  \begin{varwidth}{\linewidth}
     $j \in \mathcal{R} ({\bf z})$;\\
    $P(\theta_j \leq 0 |z_j) > P(\theta_j \geq 0 |z_j)$;\\
    $ \widehat{sgn}_j ({\bf z})= 1$.
   \end{varwidth}
  \right\rbrace
\]
have non-zero probabilities, the procedure  $ (\widehat{sgn}_{i}')_{i \in \cR'}$ defined by 
\[
\mathcal{R}'({\bf z}) = \mathcal{R}({\bf z}) \text{ for all } {\bf z}
\]
and, for $i \in \mathcal{R} = \mathcal{R}'$,
\[
\widehat{sgn}'_{i}({\bf z})  = 
  \begin{cases} 
 1   & \text{if } i = j \text{ and } {\bf z} \in \mathcal{Z}_j^{(1)} \\
 -1   & \text{if } i = j \text{ and } {\bf z} \in \mathcal{Z}_j^{(2)} \\
 \widehat{sgn}_{i}({\bf z})  & \text{if otherwise}
   \end{cases}
\]
improves upon $ (\widehat{sgn}_{i})_{i \in \cR}$.
%and/or 
%\[
%% simulate displayed equation
%\mathcal{Z}_j=\left\lbrace {\bf z}\;\middle|\;
%  \begin{varwidth}{\linewidth}
%     $j \in \mathcal{R} ({\bf z})$;\\
%    $P(\theta_j \leq 0 |z_j) < P(\theta_j \geq 0 |z_j)$;\\
%    $ \widehat{sgn}_j ({\bf z})= -1$;
%   \end{varwidth}
%  \right\rbrace
%\]
%
%\begin{equation} \label{sit1}
%P(\{ {\bf z}: P(\theta_j \leq 0 |z_j) < P(\theta_j \geq 0 |z_j), j \in \mathcal{R} \text{ and } \widehat{sgn}_j ( {\bf z}) = -1 \}) > 0, 
%\end{equation}
%and/or 
%\begin{equation}\label{sit2}
%P(\{ {\bf z}: P(\theta_j \leq 0 |z_j) > P(\theta_j \geq 0 |z_j), j \in \mathcal{R} \text{ and } \widehat{sgn}_j( {\bf z})  = 1 \}) > 0.
%\end{equation}
  \item Statement  2:  If there exist two distinct  $j,  l \in [m]$ such that the event 
  \[
\mathcal{Z}_{jl}=   \{{\bf z}: \lfsr_j < \lfsr_l , l\in \mathcal{R} \text{ and } j\not \in \mathcal{R}\}
  \] has non-zero probability, then it is possible to construct an improved procedure $(\widehat{sgn}_{i}')_{i \in \cR'}$ with the property that
 \[
\mathcal{R}' ({\bf z})  = 
  \begin{cases} 
 (\mathcal{R}({\bf z})\backslash \{l\})\cup \{j\}   & \text{if } {\bf z} \in  \mathcal{Z}_{jl} \\
   \mathcal{R}({\bf z}) &  \text{if } {\bf z} \not\in  \mathcal{Z}_{jl}  \end{cases}.
\]
\end{itemize}

Suppose both statements can be shown. Then any procedure can be improved by repeatedly applying Statement 2, until we end up with a procedure for which $P(\mathcal{Z}_{jl}) = 0$ for all $(j, l)$ pairs. We can then further improve this procedure by applying Statement 1, and end up with a procedure for which $P(\mathcal{Z}_{j}^{(1)}) = P(\mathcal{Z}_{j}^{(2)}) = 0$ for all $j$, and hence satisfying the conditions $(i)$ and $(ii)$ in \thmref{ODP}; since the ODP cannot be improved, it must  have the latter two conditions satisfied.

\begin{proof} [Proof of Statement 1] %
%
%Suppose only $\mathcal{Z}_j^{(1)}$ has non-zero probability. Moreover, since $\widehat{sgn}_{i}$ is only defined when $i \in \cR$, we extend the definition by letting $\widehat{sgn}_{i} ({\bf z}) = - \infty$ if $i \not \in \cR ({\bf z})$. %\mathcal{Z} = \{ {\bf z}: 
%%P(\theta_j \leq 0 |z_j) < P(\theta_j \geq 0 |z_j) \text{ and } \widehat{sgn}_j = -1 \}.
%%\] 
%Define $\cR' ({\bf z})= \cR ({\bf z})$ for all ${\bf z}$ and let
%
%
We write
\begin{align*}
&\ETD [ (\widehat{sgn}_{i}')_{i \in \cR'}] - \ETD[(\widehat{sgn}_{i})_{i \in \cR}]\\
&=\int \sum_{i \in \mathcal{R}'({\bf z})}  P(\widehat{sgn}'_{i}({\bf z})   = \text{sgn}(\theta_i) | {\bf z}) P({\bf z}) d{\bf z} - \int \sum_{i \in \mathcal{R}({\bf z})}  P(\widehat{sgn}_{i}({\bf z})   = \text{sgn}( \theta_i) | {\bf z}) P({\bf z}) d{\bf z} \\
&=\int_{\mathcal{Z}_j^{(1)} \cup \mathcal{Z}_j^{(2)}} \sum_{i \in \mathcal{R}'({\bf z})} [ P(\widehat{sgn}'_{i}({\bf z})   = \text{sgn}(\theta_i) | {\bf z}) -  P(\widehat{sgn}_{i}({\bf z})   = \sgn(\theta_i) | {\bf z})]P({\bf z}) d{\bf z} \\
&= \int_{\mathcal{Z}_j^{(1)} \cup \mathcal{Z}_j^{(2)}} [ P(\widehat{sgn}'_j({\bf z})   = \sgn(\theta_j) | {\bf z}) -  P(\widehat{sgn}_j({\bf z})   = \sgn(\theta_j) | {\bf z})]P({\bf z}) d{\bf z} \\
&= \int_{\mathcal{Z}_j^{(1)}}  [ P(\widehat{sgn}'_j({\bf z})   = \sgn(\theta_j) | {\bf z}) -  P(\widehat{sgn}_j({\bf z})   = \sgn(\theta_j) | {\bf z})]P({\bf z}) d{\bf z} +\\
&\quad  \int_{\mathcal{Z}_j^{(2)}}  [ P(\widehat{sgn}'_j({\bf z})   = \sgn(\theta_j) | {\bf z}) -  P(\widehat{sgn}_j({\bf z})   = \sgn(\theta_j) | {\bf z})]P({\bf z}) d{\bf z}\\
 &= \int_{\mathcal{Z}_j^{(1)}}  [ P(\theta_j > 0 |z_j) -  P(\theta_j <0 |z_j)]P({\bf z}) d{\bf z} +
 \int_{\mathcal{Z}_j^{(2)}}  [ P(\theta_j <0 |z_j)) -  P(\theta_j > 0 |z_j)]P({\bf z}) d{\bf z} > 0,
\end{align*}
where  the second and third equalities come  from the fact that $ (\widehat{sgn}_{i}')_{i \in \cR'}$ and $ (\widehat{sgn}_{i})_{i \in \cR}$ differ only on $\mathcal{Z}_j^{(1)} \cup \mathcal{Z}_j^{(2)}$ and for $j$, the fourth equality is from the disjointness of  $\mathcal{Z}_j^{(1)}$ and $\mathcal{Z}_j^{(2)}$, and the last equality is from how $\widehat{sgn}_j'$ is defined on  $\mathcal{Z}_j^{(1)}$ and $\mathcal{Z}_j^{(2)}$ as well as the independence across all $i = 1, \dots, m$.  Similarly, 
\begin{align*}
&\FDRdir[(\widehat{sgn}_{i})_{i \in \cR}]-\FDRdir[ (\widehat{sgn}_{i}')_{i \in \cR'}]  \\
&= 
 \int \frac{ \sum_{i \in \mathcal{R}({\bf z})}  P(\widehat{sgn}_{i}({\bf z})   \neq \sgn(\theta_i) | {\bf z}) }{ 1 \vee |\mathcal{R}({\bf z})|}    P({\bf z}) d{\bf z} 
-\int \frac{\sum_{i \in \mathcal{R}'({\bf z})}  P(\widehat{sgn}'_{i}({\bf z})   \neq \sgn(\theta_i) | {\bf z}) }{1 \vee |\mathcal{R}'({\bf z})|}P({\bf z}) d{\bf z}
\\
&= 
 \int_{\mathcal{Z}_j^{(1)}}  \left[ \frac{   P(\widehat{sgn}_j({\bf z})   \neq \sgn(\theta_j) | {\bf z}) }{ 1 \vee |\mathcal{R}({\bf z})|}  
- \frac{  P(\widehat{sgn}'_j({\bf z})   \neq \sgn(\theta_j) | {\bf z}) }{1 \vee |\mathcal{R}'({\bf z})|}\right] P({\bf z}) d{\bf z} + \\
&\quad \int_{\mathcal{Z}_j^{(2)}}  \left[ \frac{   P(\widehat{sgn}_j({\bf z})   \neq \sgn(\theta_j) | {\bf z}) }{1 \vee  |\mathcal{R}({\bf z})|}  
- \frac{  P(\widehat{sgn}'_j({\bf z})   \neq \sgn(\theta_j) | {\bf z}) }{1 \vee |\mathcal{R}'({\bf z})|}\right] P({\bf z}) d{\bf z}\\
&= 
 \int_{\mathcal{Z}_j^{(1)}}  \left[ \frac{   P(\theta_j \geq 0 |z_j) }{ 1\vee |\mathcal{R}({\bf z})|}  
- \frac{  P(\theta_j \leq 0 |z_j) }{1 \vee |\mathcal{R}'({\bf z})|}\right] P({\bf z}) d{\bf z} + 
 \int_{\mathcal{Z}_j^{(2)}}  \left[ \frac{   P(\theta_j \leq 0 |z_j) }{1 \vee | \mathcal{R}({\bf z})|}  
- \frac{   P(\theta_j \geq 0 |z_j) }{1 \vee |\mathcal{R}'({\bf z})|}\right] P({\bf z}) d{\bf z} \geq 0.
\end{align*} 
so $ (\widehat{sgn}'_{i})_{i \in \cR'}$ improves upon  $(\widehat{sgn}_{i})_{i \in \cR}$. 
\end{proof}
%If only $\mathcal{Z}_j^{(2)}$ has non-zero probability, one can define the  improved procedure $\{\cR', (\widehat{sgn}_{i}')_{i \in \cR'}\Bigr\}$ completely analogously. Moreover, when both $\mathcal{Z}_j^{(1)}$ and $\mathcal{Z}_j^{(2)}$ have non-zero probabilities, since the two sets are disjoint an improved procedure can obviously also be formed. Hence Case 1 is addressed.

\begin{proof} [Proof of Statement 2]

The proof for Statement 2 follows in a similar vein.  We first define 
\[
% simulate displayed equation
\mathcal{Z}_{jl}^{(1)}\equiv\left\lbrace {\bf z}\;\middle|\;
  \begin{varwidth}{\linewidth}
${\bf z} \in \mathcal{Z}_{jl}$;\\
    $P(\theta_j \leq 0 |z_j) < P(\theta_j \geq 0 |z_j)$.
   \end{varwidth}
  \right\rbrace
  and \ \
  % simulate displayed equation
\mathcal{Z}_{jl}^{(2)}\equiv\left\lbrace {\bf z}\;\middle|\;
  \begin{varwidth}{\linewidth}
${\bf z} \in \mathcal{Z}_{jl}$;\\
    $P(\theta_j \leq 0 |z_j) > P(\theta_j \geq 0 |z_j)$.
   \end{varwidth}
  \right\rbrace.
\]
Since $\mathcal{R}'$ has been defined, we only have to define $\widehat{sgn}'_{i}$ for each $i \in \mathcal{R}'$, as
\[
\widehat{sgn}'_{i}({\bf z})  = 
  \begin{cases} 
 1   & \text{if } i = j \text{ and } {\bf z} \in \mathcal{Z}_{jl}^{(1)} \\
 -1   & \text{if } i = j \text{ and } {\bf z} \in \mathcal{Z}_{jl}^{(2)} \\
 \widehat{sgn}_{i}({\bf z})  & \text{if otherwise}
   \end{cases}.
\]
One can then write 
\begin{align}
&\ETD[ (\widehat{sgn}'_{i})_{i \in \cR'}] - \ETD[(\widehat{sgn}_{i})_{i \in \cR}] \notag\\
&= \int_{\mathcal{Z}_{jl}^{(1)}}  [ P(\widehat{sgn}'_j({\bf z})   = \text{sgn}(\theta_j) | {\bf z}) -  P(\widehat{sgn}_l({\bf z})   = \text{sgn}(\theta_l) | {\bf z})]P({\bf z}) d{\bf z} \quad  + \notag\\
&\quad 
 \int_{\mathcal{Z}_{jl}^{(2)}}  [ P(\widehat{sgn}'_j({\bf z})   = \sgn(\theta_j) | {\bf z}) -  P(\widehat{sgn}_l({\bf z})   = \sgn(\theta_l) | {\bf z})]P({\bf z}) d{\bf z} \notag\\
 &= \int_{\mathcal{Z}_{jl}^{(1)}}  [ P(\theta_j > 0|{\bf z}) -  P(\widehat{sgn}_l({\bf z})   = \text{sgn}(\theta_l) | {\bf z})]P({\bf z}) d{\bf z} \quad + \label{finaltag1}\\
 &\quad \int_{\mathcal{Z}_{jl}^{(2)}}  [ P(\theta_j < 0|{\bf z}) -  P(\widehat{sgn}_l({\bf z})   =\text{sgn}( \theta_l) | {\bf z})]P({\bf z}) d{\bf z}. \label{finaltag2}
\end{align}
by a similar train of equalities as in the proof for Statement 1. For ${\bf z} \in \mathcal{Z}_{jl}^{(1)}$, $P(\theta_j \leq 0 |{\bf z}) = \lfsr_j < \lfsr_l \leq P(\widehat{sgn}_l({\bf z})   \neq \text{sgn}(\theta_l) | {\bf z})$, where the last inequality comes from the fact that $\lfsr_l = P( \theta_l \leq 0 |{\bf z})\wedge P( \theta_l \geq 0 |{\bf z})$ is the smallest conditional probability of making a false discovery that can possibly be achieved by $\widehat{sgn}_l ({\bf z})$. This in turns implies  $ P(\theta_j > 0 |{\bf z})  > P(\widehat{sgn}_l({\bf z})   = \text{sgn}(\theta_l) | {\bf z})$, which means that the term in \eqref{finaltag1} is greater than 0. Similarly one can show that the term in \eqref{finaltag2} is also greater than 0, which gives $\ETD[ (\widehat{sgn}'_{i})_{i \in \cR'}] - \ETD[(\widehat{sgn}_{i})_{i \in \cR}] > 0$.

We can also show that $\FDRdir [ (\widehat{sgn}'_{i})_{i \in \cR'}] \leq \FDRdir[(\widehat{sgn}_{i})_{i \in \cR}]$ similarly; to avoid repetitions, we leave the details to the reader. 
\end{proof}

\newpage
% new appendix
\section{Additional content for \secref{simul}} 

\subsection{ Implementation of the LSFR  procedure} \label{app:imple}
An exact implementation of the ODP described in \thmref{ODP}  involves solving a rather complex infinite integer programming problem  \citep{heller2021optimal} to determine a threshold for the local false sign rates. As an alternative, in \secref{simul},  LFSR is  a similar  oracle procedure with an attractively  simpler implementation first suggested by \citet{sun2007oracle}, and it suffices to serve as an oracle benchmark for the power of our compared procedures. Suppose we denote this procedure as $(\widehat{sgn}_i^{SC})_{i \in \mathcal{R}_{SC}}$  in our notation. Then the signs $\widehat{sgn}_i^{SC}$ are  declared as in \tabref{OptSign}, i.e., $\widehat{sgn}_i^{SC} = \widehat{sgn}_i^{ODP}$, and the discovery set  is defined by 
\[
 \mathcal{R}_{SC}\equiv \{i: \lfsr_i \leq \lfsr_{(j)}\}
\]
with the index 
  \begin{equation*}\label{oraclej}
j = j(q) \equiv \max \left\{ i' \in  [m] :  \frac{\sum_{i = 1}^{i'}  \lfsr_{(i)} }{i'}  \leq q \right\},
\end{equation*}
where  $\lfsr_{(1)} \leq  \dots \leq \lfsr_{(m)}$ are the order statistics of true local false sign rates, and $ \mathcal{R}_{SC}$ is the empty set if $j$ is not well-defined. The ratio $\frac{\sum_{i = 1}^{i'}  \lfsr_{(i)} }{i'}$ in the definition of $j$ is precisely the conditional $\FDRdir$
\[
\mathbb{E}\left[ \frac{|\{i: \widehat{sgn}_i^{SC} \neq \text{sgn}(\theta_i) \text{ and } \lfsr_i \leq \lfsr_{(i')} \}|}{i'} \Big| \{z_i\}_{i \in [m]}\right]
\]
 of  optimally declaring the signs for the subset $\{i: \lfsr_i \leq \lfsr_{(i')}\}$ given the data, which also implies the $\FDRdir$ of $(\widehat{sgn}_i^{SC})_{i \in \mathcal{R}_{SC}}$,
 \begin{multline*}
 \mathbb{E} \left[   \frac{\widehat{sgn}_i^{SC} \neq \text{sgn}(\theta_i) }{ 1 \vee |\mathcal{R}_{SC}|} \right] = 
  \mathbb{E} \left[   \mathbb{E} \left[   \frac{\widehat{sgn}_i^{SC} \neq \text{sgn}(\theta_i) }{ 1 \vee |\mathcal{R}_{SC}|}\Big| \{z_i\}_{i \in [m]} \right]\right] \\
=    \mathbb{E} \left[  \frac{\sum_{i = 1}^{j}  \lfsr_{(i)} }{j} \Big|   \mathcal{R}_{SC} \neq \varnothing \right] P(\mathcal{R}_{SC} \neq \varnothing),
 \end{multline*}
 is less than $q$ by how $j$ was defined, under the Bayesian formulation \eqref{prior}.

\subsection{Other simulation results} \label{app:simResultsExtra}

An R package for our methods is available at 
\url{https://github.com/ninhtran02/zdirect}. We also conducted additional simulation studies to evaluate the performance of different methods under dependent $z$-values. Specifically, given $\btheta = (\theta_1, \dots, \theta_m)$ generated exactly as described in \secref{simulation_setup_prototype} using the same simulation parameters, we generate $\bz = (z_1, \dots, z_m)$ with a multivariate normal distribution $N(\btheta, \Sigma)$ and an autoregressive covariance structure
$
\Sigma_{ij} = \rho^{|i -j|}
$ for $1 \leq i, j, \leq m$. The following values of $\rho$ are experimented with:

\subsubsection{Positive autoregressive dependence} \label{app:positive_ar1}
Weak and strong positive dependence with $\rho = 0.5$ and $\rho = 0.8$.

\subsubsection{Negative autoregressive dependence}  \label{app:negative_ar1}
 Weak and strong negative dependence with $\rho = - 0.5$ and $\rho = -  0.8$.

\subsubsection{Brief summary}  \label{app:summary_additional_simulations}

The performances of different methods are included in \figsref{result_ar1_small_positive}  and \figssref{result_ar1_large_positive} for the positively dependent settings in \appref{positive_ar1}, and \figsref{result_ar1_small_negative} and \figssref{result_ar1_large_negative} for the negatively dependent settings in  \appref{negative_ar1}. Note that, in addition to the existing methods  from \secref{methods_compared_sim},  we have included an extra method called ``$\dBHdir$'' in our results. The term ``dBH'' refers to the \emph{dependence-adjusted BH procedure}, a recent advancement in FDR testing under arbitrary dependence proposed by \citet{dBH_paper}. It serves as a theoretically valid and more powerful alternative to the widely recognized but very conservative BY procedure \citep{benjamini2001control}. We note that dBH requires prior knowledge of the underlying dependence structure, distinguishing it from the BY procedure. $\dBHdir$ is a variant of dBH designed specifically for multivariate normal $z$-values and $\FDRdir$ control. It is implemented through the function \texttt{dBH\_mvgauss} in the R package \texttt{dbh}, where the \texttt{gamma} parameter is set to $0.9$ following the recommendation by \citet{dBH_paper} for two-sided testing. Its exact $\FDRdir$ control for our settings with dependent $z$-values  is established by \citet[Corollary 7]{dBH_paper}. 

The $\FDRdir$ and power of each method under the dependent settings outlined in \appref{positive_ar1} and \appref{negative_ar1} are similar to those under the independent setting outlined in \secref{simul}. Nevertheless, subtle differences emerge. For strong autoregressive dependence $\rho \in \{-0.8,0.8 \}$, when $w = 0.8$ and $\xi \in \{ 0.5, 1\}$, the methods $\text{STS}_{\text{dir}}$, $\text{aSTS}_{\text{dir}}$, $\BHdir$, LFSR, GR and $\dBHdir$ exhibited slight $\FDRdir$ decreases by approximately $0.01$ to $0.02$. Conversely, ZDIRECT displayed slight $\FDRdir$ increases by approximately $0.01$. Despite these increases in $\FDRdir$ under strong autoregressive dependence, ZDIRECT consistently maintained empirical control of $\FDRdir$ below the designated level of $q = 0.10$ throughout our additional simulations.

We observed that $\dBHdir$ is slightly more conservative in $\FDRdir$ control and less powerful than $\BHdir$ across our additional simulations. This difference in performance may be attributed to the recommended \texttt{gamma} parameter choice of $0.9$ by \cite{dBH_paper}, chosen to reduce the likelihood of obtaining a randomly ``pruned'' rejection set; see \citet[Section 2.2]{dBH_paper} for an explanation of why it is preferable to avoid the randomized pruning step built into dBH-type methods. This cautious parameter choice may compromise any potential power advantage $\dBHdir$ could have over $\BHdir$ in the presence of autoregressive dependence.

\begin{figure*}[htp]
\centering
\includegraphics[scale=0.66]{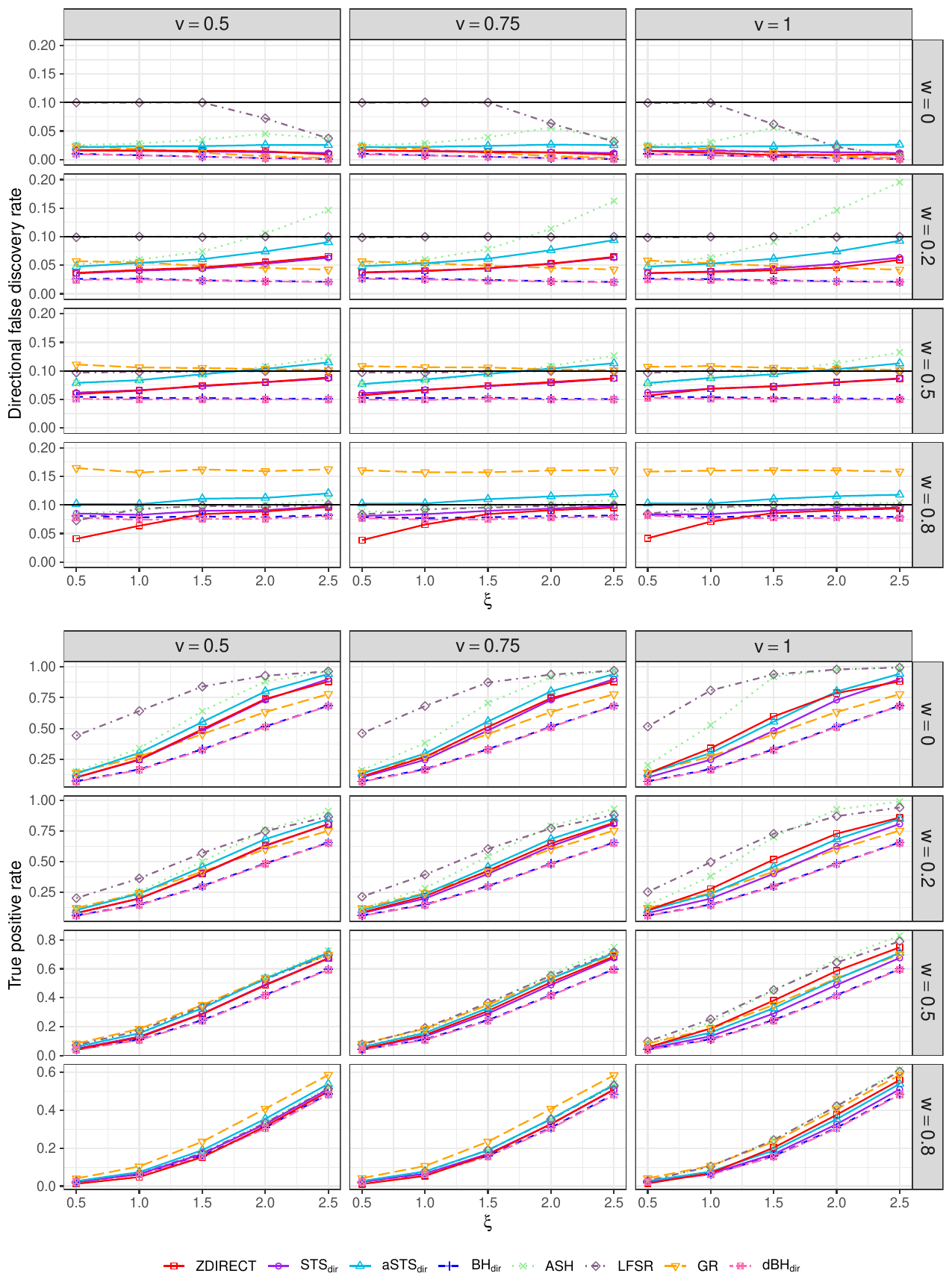}
\caption{
Empirical directional false discovery rate and true positive rates of the eight compared methods for the simulations in \appref{positive_ar1} with $\rho = 0.5$; each method was implemented at a target $\FDRdir$ level $q =0.1$ (black horizontal lines).
 }
\label{fig:result_ar1_small_positive}
\end{figure*}

\begin{figure*}[htp]
\centering
\includegraphics[scale=0.66]{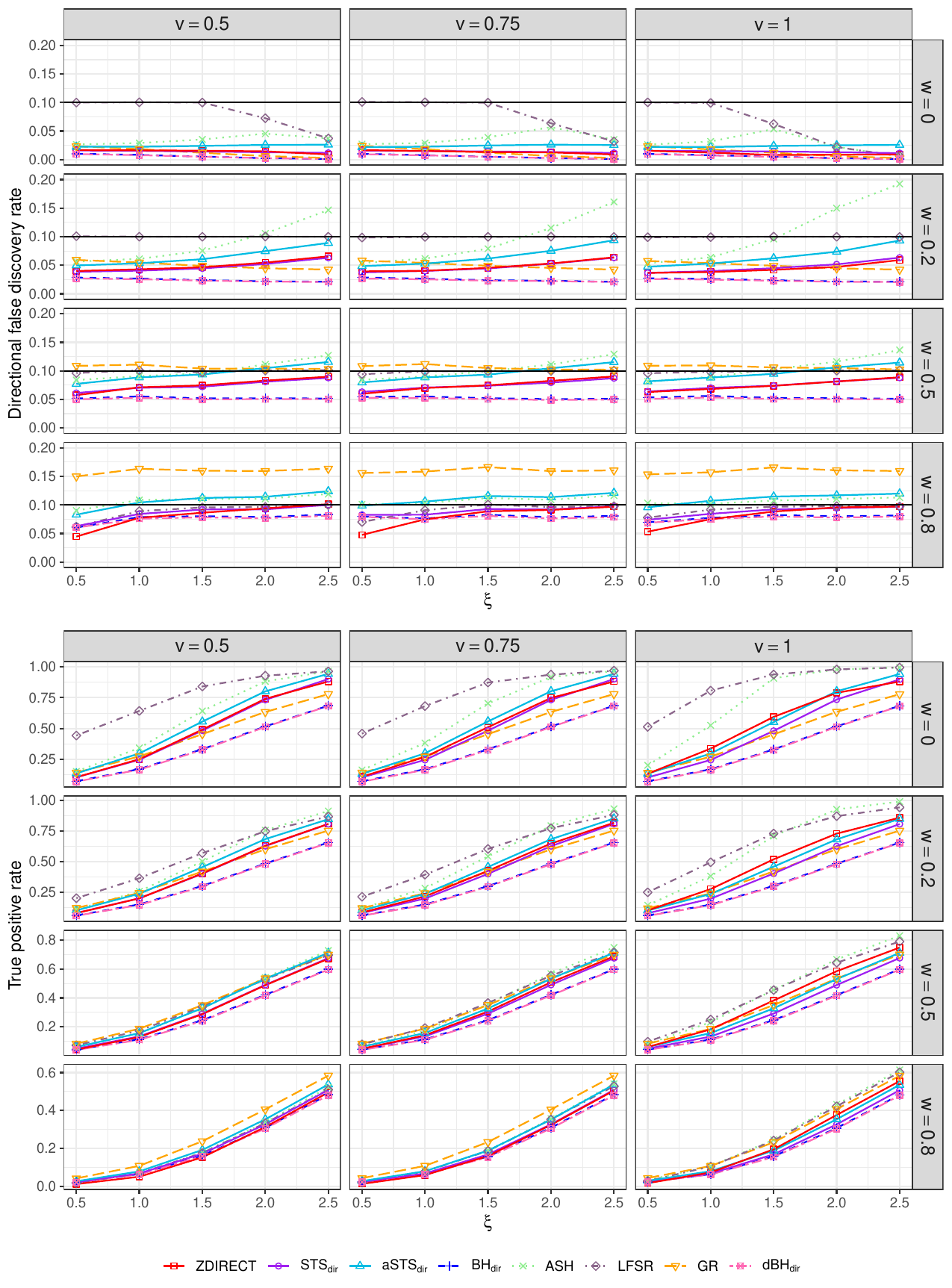}
\caption{
Empirical directional false discovery rate and true positive rates of the eight compared methods for the simulations in \appref{positive_ar1} with $\rho = 0.8$; each method was implemented at a target $\FDRdir$ level $q =0.1$ (black horizontal lines).
 }
\label{fig:result_ar1_large_positive}
\end{figure*}

\begin{figure*}[htp]
\centering
\includegraphics[scale=0.66]{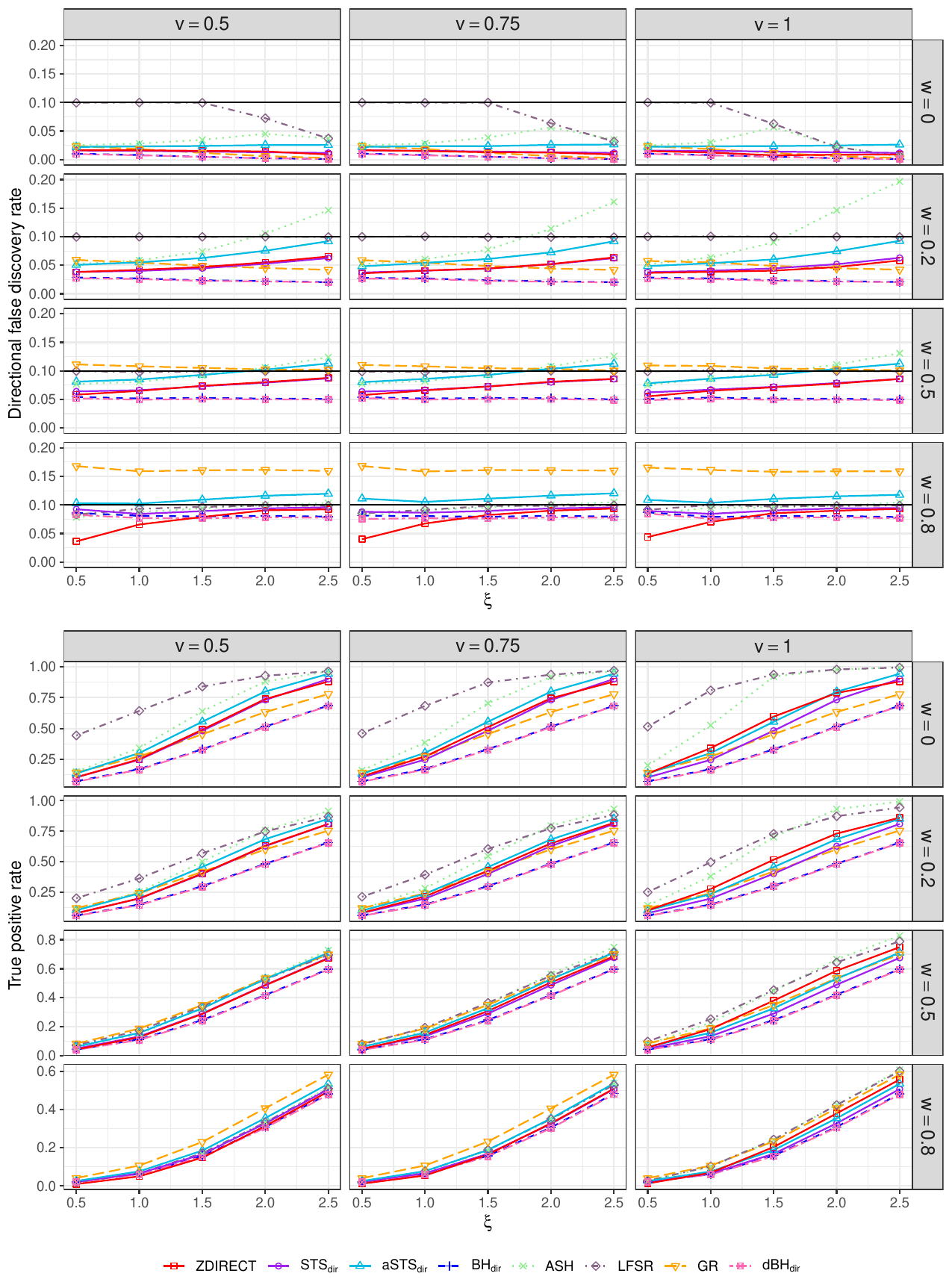}
\caption{
Empirical directional false discovery rate and true positive rates of the eight compared methods for the simulations in \appref{negative_ar1} with $\rho = - 0.5$; each method was implemented at a target $\FDRdir$ level $q =0.1$ (black horizontal lines).
 }
\label{fig:result_ar1_small_negative}
\end{figure*}

\begin{figure*}[htp]
\centering
\includegraphics[scale=0.66]{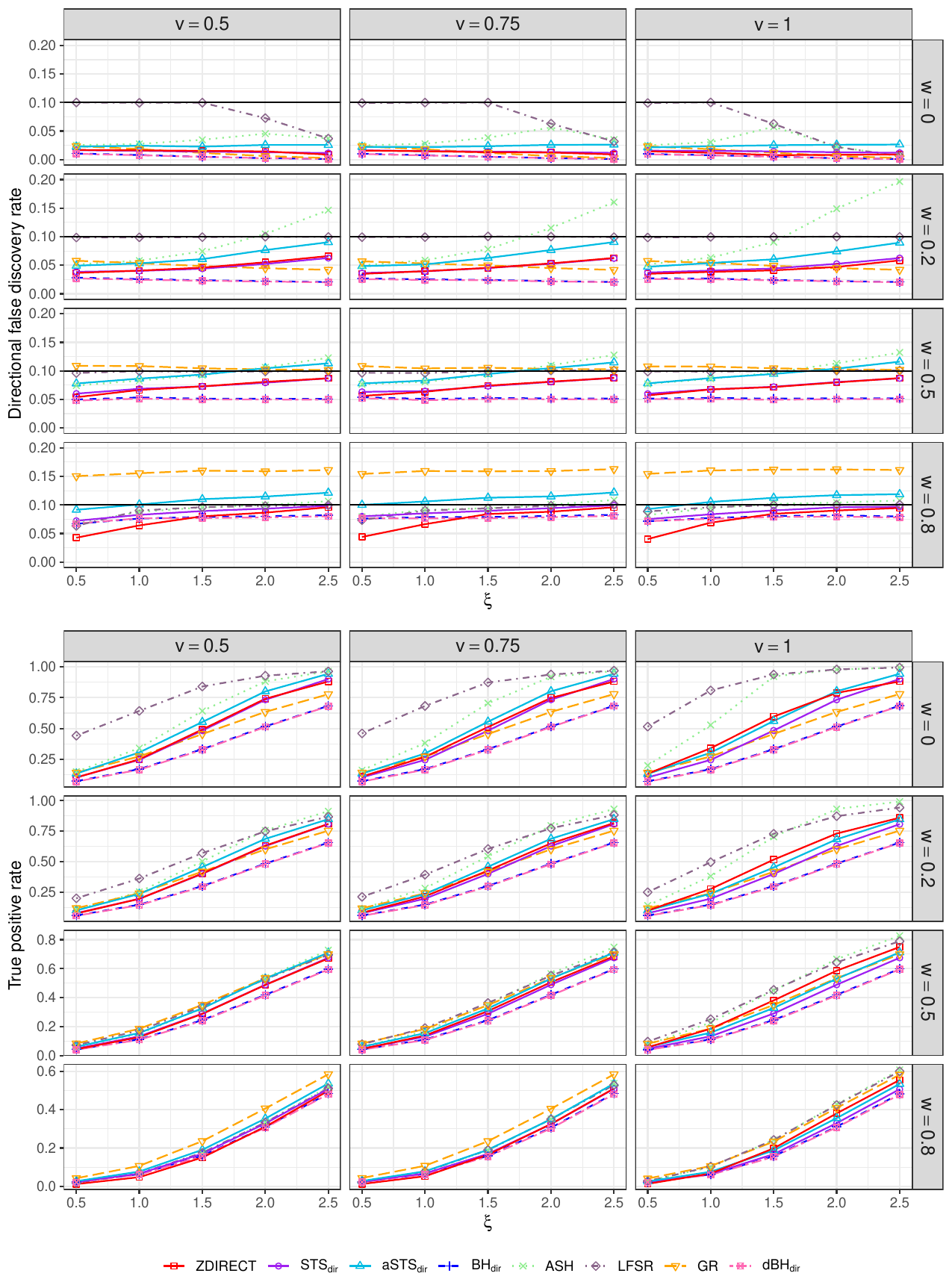}
\caption{
Empirical directional false discovery rate and true positive rates of the eight compared methods for the simulations in \appref{negative_ar1} with $\rho = - 0.8$; each method was implemented at a target $\FDRdir$ level $q =0.1$ (black horizontal lines).
 }
\label{fig:result_ar1_large_negative}
\end{figure*}

\end{appendix}

%%%%%%%%%%%%%%%%%%%%%%%%%%%%%%%%%%%%%%%%%%%%%%
%% Acknowledgements                         %%
%% should be provided in the                %%
%% Acknowledgements section.                %%
%%%%%%%%%%%%%%%%%%%%%%%%%%%%%%%%%%%%%%%%%%%%%%
\begin{acks}[Acknowledgments]
The authors would like to thank the anonymous referees and an Associate
Editor  for their constructive comments that improved the
quality of this paper.
\end{acks}

\bibliographystyle{imsart-nameyear.bst} % Style BST file (imsart-number.bst or imsart-nameyear.bst)
\bibliography{dFDRref.bib}       % Bibliography file (usually '*.bib')

\end{document}